\documentclass{vldb}
\usepackage{cite}
\usepackage{graphicx}
\usepackage{amsmath}
\usepackage{url}
\usepackage[noend, linesnumbered, boxed, ruled]{algorithm2e}
\usepackage[noend]{algpseudocode}
\usepackage{esvect}
\usepackage{mathtools}
\usepackage{epstopdf}
\usepackage{subcaption}
\usepackage{tikz}
\usepackage{pstricks}
\usepackage{float}
\usepackage{amsfonts}
\usepackage{paralist}
\usepackage{balance}
\usepackage[hyperindex,breaklinks]{hyperref}

\PassOptionsToPackage{linktocpage}{hyperref}

\begin{document}
%
\title{Integrative Dynamic Reconfiguration in a \\Parallel Stream Processing
Engine}

\numberofauthors{3}
\author{
\alignauthor Kasper Grud Skat Madsen\\
       \affaddr{University of \\Southern Denmark}\\
       \email{kaspergsm@imada.sdu.dk}
\alignauthor Yongluan Zhou\\
       \affaddr{University of \\Southern Denmark}\\
       \email{zhou@imada.sdu.dk}
\alignauthor Jianneng Cao\\
       \affaddr{Institute for Infocomm \\Research in Singapore}\\
       \email{caojn@i2r.a-star.edu.sg}
}

\date{1 February 2016}

\maketitle

\begin{abstract}
Load balancing, operator instance collocations and horizontal scaling are
critical issues in Parallel Stream Processing Engines to achieve low data
processing latency, optimized cluster utilization and minimized communication
cost respectively.  In previous work, these issues are typically tackled
separately and independently. We argue that these problems are tightly coupled
in the sense that they all need to determine the allocations of workloads and
migrate computational states at runtime.  Optimizing them independently would
result in suboptimal solutions.  Therefore, in this paper, we investigate how
these three issues can be modeled as one integrated optimization problem. In
particular, we first consider jobs where workload allocations have little effect
on the communication cost, and model the problem of load balance as a
Mixed-Integer Linear Program. Afterwards, we present an extended solution called
ALBIC, which support general jobs.  We implement the proposed techniques on top
of Apache Storm, an open-source Parallel Stream Processing Engine. The extensive
experimental results over both synthetic and real datasets show that our
techniques clearly outperform existing approaches.
  

%
%
\end{abstract}

\section{Introduction}
Recently, Parallel Stream Processing Engines (PSPEs) are emerging to process the
ever growing big streams of data generated by mobile devices,
sensors, online social networks, online financial transactions, and so on.
Representatives of modern PSPEs include Apache Storm
\cite{Toshniwal:2014:STO:2588555.2595641}, Apache S4 \cite{Neumeyer:S4} and
Google MillWheel \cite{41378}. A job in an PSPE consists of a set of operators,
and each operator is usually parallelized into a number of instances, each
processing a subset of the operator's input to increase the data throughput.
The output of the instances of upstream operators, except the sink operators,
is input to their downstream neighbors, forming a pipelined network, called the
topology of the job.


In general, there are three highly correlated job optimization problems in a PSPE.
First of all, load balancing across the processing nodes is critical to
maximizing cluster utilization and minimizing data processing latency especially
during temporary load spikes.  While load balancing in parallel computing
systems~\cite{Overeinder1996101, Ross:1991:OLB:116825.116847} has been studied
extensively in the literature, developing load balancing mechanisms in PSPEs
faces new challenges. Operator instances in PSPEs are long-standing, and
probably need to be reallocated many times during their life-time. In other
words, load rebalancing decisions have to be made continuously and periodically
to maintain satisfactory system performance. As operator instances in PSPEs may
be associated with computing states and their reallocations may involve state
migrations with significant overheads, load balancing decisions have to take
such overheads into account.


Secondly, neighboring operator instances in PSPEs continuously transfer data
between each other. Collocating them at the same node will significantly reduce
the system's workload by eliminating the overhead of cross-node data
transmission, which includes both network bandwidth
consumption~\cite{Pietzuch:2006:NOP:1129754.1129909,DBLP:journals/debu/AhmadCJZZ05}
and CPU consumption caused by data serialization and deserialization.  Moreover,
the fact that the load of an operator instance is correlated with its relative
location to its neighbors also complicates the problem of load balancing. 
Assigning a location-independent load value to each instance cannot model the
real situation and hence cannot optimize the performance. To simplify the
problem, previous studies of dynamic load balancing in PSPEs, such
as~\cite{shah:2003:flux, Xing:2005:DLD:1053724.1054117, zhou2006efficient,
7113279} largely ignore the effect of collocating operator instances.



Thirdly, with the development of cloud computing platforms and distributed
system kernels (such as Mesos~\cite{hindman2011mesos}), horizontal
scaling of an PSPE, i.e.  the ability to dynamically acquire and relinquish
computing resources at runtime, is crucial to achieve high resource utilization
and low operational cost. We argue that the problem of horizontal scaling is
tightly coupled with both load balancing and operator instance collocation. For
example, the overload of the system could possibly be rectified by collocating
operator instances that have a high communication volume between each other to
save the cost of data serialization and deserialization instead of acquiring
more resources. Another example is that horizontal scaling, load balancing and
operator instance collocation all involve state migrations, which may incur
significant overhead. Hence optimizing them integratively would result
in a solution with better load balancing and lower overhead. However, there is
a lack of study on how to perform such an integrative optimization.

To fill the gap, we revisit the load balancing problem in PSPEs and propose a
novel solution that also takes operator instance collocation and horizontal
scaling into account.  We first study the scenarios where collocating
operator instances has little effect. Such scenarios occur when each operator
instance has to transfer data to a lot of neighbors evenly.
We model this problem as a \textit{Mixed-Integer Linear Program} (MILP).
Then we extend our solution to cases where collocating operator instances can
significantly affect the system performance. The extended solution, called ALBIC
(Autonomic Load Balancing with Integrated Collocation), dynamically constrains
the MILP so that it gradually improves the collocation at runtime while ensuring
a user-defined load-balance constraint is met.



%

In
summary, the contributions of this paper include:



\begin{compactitem}
  \item We propose a simple yet effective adaptation framework that integrates
	the dynamic optimization of load balancing, operator instance allocation and
	horizontal scaling. To the best of our knowledge, this is the first work
	toward this direction.

  \item We model load balancing and horizontal scaling as an integrated
	Mixed-Integer Linear Program (MILP), and by using an LP solver, achieve a
	load distribution with much better balance than existing heuristic
	approaches.

	
  \item We propose ALBIC, which automatically detects beneficial instance
	collocations and gradually increases such collocations, while maintaining a
	good load balance.
  
\item We implement all the proposed techniques on top of Apache
  \mbox{Storm}~\cite{Toshniwal:2014:STO:2588555.2595641} and conduct extensive
  experiments using both synthetic and real datasets to examine the performance
  of the proposed approaches and compare them with the state-of-the-art
  approaches~\cite{shah:2003:flux, 7113279, khandekar2009cola}. The results show
  that our integrative approaches significantly outperform the existing ones.  

\end{compactitem}

\section{Related Work}
\label{sec:relatedWork}

\subsection{Static Scheduler}
%

SODA~\cite{wolf2008soda} adopts a static job admission algorithm based on the
capacity of the cluster and then allocates operators to the nodes, using a
simple heuristic approach. SQPR~\cite{Kalyvianaki:2011:SSQ:2004686.2005583}
optimizes query admission, operator allocation, query reuse and load balancing
as an integrative problem. FUGU~\cite{heinze2013elastic} employs load balancing
to support horizontal scaling when adding or removing jobs.

COLA~\cite{khandekar2009cola} uses a balanced graph partitioning algorithm to
put operators into a number of balanced partitions to achieve both load
balancing and minimization of cross-node communication. It first puts all
operators into one partition, and then gradually splits the partitions until a
solution with a sufficient load balance is obtained. Splitting is done by a
balanced graph partitioning algorithm, which ensures the load of each partition
is relatively even, while minimizing the data sent across different partitions.


Stanoi et al.~\cite{4288141} focus on operator allocations that maximize data
throughput. Rivetti et al.~\cite{Rivetti:2015:EKG:2675743.2771827} focus on
optimizing load balance. Lastly, several works
\cite{Xing:2006:PRL:1182635.1164194, Lei:2014:RDQ:2627748.2602138} attempt to
allocate operators such that the system becomes more resilient to
workload-fluctuations and hence the overhead of state migrations can be avoided.

\subsection{Adaptive Scheduler}

{\bf Operator placement without intra-operator parallelization.} Many early
researches assume an operator can be processed by a single processing
node, and hence they do not consider intra-operator parallelization. Xing et al.
\cite{Xing:2005:DLD:1053724.1054117} and Zhou et al.  \cite{zhou2006efficient}
have studied how to dynamically allocate operators within a cluster to achieve
load balancing. A more recent paper~\cite{Gounaris:2012:ELB:2168260.2168265}
improves the load balancing solution by using a multi-input multi-output
feedback linear quadratic regulator
based on control theory.
Jian et al. \cite{4812453} optimize 
operator allocation to minimize the total communication cost with
a large number of jobs.
Pietzuch et al. \cite{Pietzuch:2006:NOP:1129754.1129909} and Yanif et al.
\cite{DBLP:journals/debu/AhmadCJZZ05} consider operator allocation that
minimizes network usage between operators. As all the above work do not consider
intra-operator parallelization, they cannot easily be adopted for solving our
problem.

{\bf Load balancing with intra-operator parallelization.} \\
Flux~\cite{shah:2003:flux} is one of the very few early researches that consider
intra-operator parallelization. It adapts data partitioning periodically. 
At the end of each period, nodes are sorted in descending order of their
workloads, and then it moves the biggest suitable data partition at the first node to the
last one in the list, so that load-variance is decreased. If necessary,
it also moves a partition at the second node to the
second last one in the list, and so on.
Spark Streaming~\cite{Zaharia:2013:DSF:2517349.2522737} processes input as a
series of mini-batches, where a mini-batch can be processed by a task on any
node. Load balancing can be easily achieved by allocating tasks. Both these
methods do not consider collocating operator instances.

Several recent works focuses on developing partitioning functions that can
achieve load balancing without or with few state migrations. A recent
paper~\cite{gedik2014partitioning} discusses how to define a partitioning
function, which can balance load, memory and bandwidth, while ensuring changes to the
partitioning function impacts as small a subset of keys as possible, in order to
minimize the need for state migration. The method of ``The Power of Two
Choices'' (PoTC)~\cite{7113279} continuously defines two hash functions $h_1(x)$
and $h_2(x)$, such that each key $x$ can be sent to one of two alternative
downstream operator instances. Each operator instance tries to balance the
amount of work sent downstream, such that all operator instances downstream
receives an even workload. Since the state is split over two operator instances
per key, the state must be merged before the final computation can be applied.
This incurs a continuous overhead, even if no load balancing is needed for the
job. Notice that the merge step cannot be balanced, which can lead to unbalanced
workload in case the merge step is costly. Again this line of works also do not
consider collocating operator instances.

\subsection{Dynamic Scaling}
Apache Mesos \cite{hindman2011mesos} is a platform that facilitates sharing of
computational resources among the (different) systems on a cluster, which can
help maximize its utilization. 
To benefit from a platform like Mesos, a PSPE has to adaptively reconfigure its
resource usage according to its present workload.

There are multiple PSPEs~\cite{satzger:esc, 7113328, 6008710,
ElasticScalingForDataStreamProcessing} that support horizontal scaling and load
balancing, by collecting system-wide statistics over a certain timespan, then
making global decisions on the basis of these.
However, these techniques do not consider collocation of operator instances to
reduce communication overhead.

StreamCloud~\cite{Gulisano:2012:SES:2412378.2412792} supports both dynamic
scaling and load balancing, and uses a simplified and static operator
collocation method. It assumes a job is composed by relational operators,
whose semantics are known to the system. They group a stateless operator, which
can be parallelized randomly (e.g. a selection operator), with a neighboring
stateful operator, which should be parallelized by the key of the input (e.g. a
group-by aggregate operator), into a component. Then the component
will be parallelized using the key of the stateful operator. This approach
achieves collocation of some communicating instances, but not the
instances of two stateful operators. Our approach, on the other hand, assumes
the operator semantics are opaque to the system and considers collocations of
instances of any operators.



There also exist work mainly focused on dynamically calculating the number of
resources needed for a stream processing job with various goals, such as
minimizing the monetary cost of using cloud services~\cite{6008710}, and
limiting the expected processing latency~\cite{DBLP:journals/corr/FuDMWYZ15}.
Another recent paper~\cite{Li:2015:SSA:2809974.2809979} considers both latency
and throughput in a holistic manner. Based on a user specified latency constraint, their approach will determine the degree of parallelism and
granularity of scheduling for the computation, in order to satisfy the latency
constraint while achieving maximum throughput. These methods can be adopted in
our framework to calculate the amount of needed resources when making horizontal
scaling decisions.

\section{System Model}
\label{sec:preliminaries}

\textbf{Data Model.}
Input data of each operator is modelled as a number of continuous streams of
tuples in the form of $\langle key, value, ts \rangle$, where $key$ is used to
partition the operator's input stream, $value$ is content of the tuple, and
$ts$ is its timestamp. Both $key$ and $value$ are opaque to the system.

\textbf{Query Model.}
A job corresponds to a set of user-defined continuous queries. The queries can
be formulated as an operator network, which is a directed acyclic graph $\langle
O, E\rangle$, where each vertex is an operator $O_i$ and each edge is a stream,
where the direction represents the direction of data flow. The
\textit{src} operators produce inputs for the job and the \textit{sink}
operators produce no output. By allowing a sink operator being attached to any
operator, the query model supports concurrent queries and operator sharing.


\begin{table}
\center
\begin{tabular}{p{2.4cm} p{5.2cm}}
\bf{Symbol} & \bf{Description} \\
\hline
$n_i$ & Node $i$\\
$O_i$ & Operator $i$\\
$o_j$ & Operator instance $j$\\
$g_k$ & Key group $k$\\
$\sigma_k$ & Computation state of $g_k$\\
$load_i$ & Load of node $n_i$\\
$gLoad_k$ & Load of $g_k$\\
$kill_i$ & Binary variable indicating if $n_i$ is marked for removal\\
$SPL$ & Statistics Period Length\\
\end{tabular}
\caption{Symbols}
\label{tab:symboloverview}
\end{table}

\textbf{Execution Model.}
Each input tuple of $O_i$ is associated with a key and these input keys
are partitioned into a number of non-overlapping subsets, each is called a key
group and denoted as $g_k$. The main assumption is that the processing of key
groups is independent of one another. Each key group $g_k$ must therefore have
an independent processing state $\sigma_k$. 


A cluster has a set of nodes $N = \{n_1,\ldots,n_{|N|}\}$ and each node $n_i$
processes a non-overlapping subset of key groups from any operator. If key groups
$g_k$ and $g_l$ are both processed at node $n_i$, they are said to be
collocated. If a subset of key groups from operator $O_j$ is allocated at
$n_i$, we say that $n_i$ possesses an operator instance $o_i$ of $O_j$.  For
simplicity we also use $O_i$ to denote the set of instances of $O_i$.

\textbf{Processing Order.}
Obtaining a reproducible ordering of input tuples is costly, because an operator
can have several unsynchronized input streams from multiple upstream operator
instances. In this paper, we assume the system employs out-of-order processing
techniques, which means that operators eventually produce the same result for
the same input data as long as the unorderedness is within some bound
\cite{Li:2008:OPN:1453856.1453890}. Some computation relies on processing the
input in a predefined strict order, in which case one has to order the data
using an additional function before feeding the data to the actual computation,
e.g. the SUnion function \cite{Balazinska:2008:FBD:1331904.1331907}. Therefore,
our assumption of out-of-order processing does not exclude applications that
require a strict input order.

\textbf{State Migration.} State Migration is conducted using \textit{direct
state migration}~\cite{Madsen:2015:DRM:2806416.2806449}, which works as follows.
Consider moving one key group $g_k \in O_i$ from node $n_1$ to
$n_2$. Initially, all the instances of upstream operators of
$O_i$ are informed that they must redirect new tuples for key group $g_k$ to
$n_2$. All the new tuples are then buffered at $n_2$. Node $n_1$ serializes all
the data, which is necessary to move the key group $g_k$, and sends it to $n_2$.
Lastly, $n_2$ deserializes the data from $n_1$ to re-construct key group $g_k$,
and processes all the buffered data. The state migration cost is calculated by
the cost model given in the referenced paper.  All the solutions proposed in
this paper, are independent on the actual state migration technique applied,
hence alternative state migration techniques~\cite{7113328,
Madsen:2015:DRM:2806416.2806449, CastroFernandez:2013:ISO:2463676.2465282} can
be applied.

\textbf{Statistics.}
The system maintains statistics on the usages of CPU, memory and network bandwidth, as
well as the input- and output data rates of processing each key group. 
Such statistics are collected and calculated over every period $P_{i \rightarrow
j}: [T_i,T_{j}]$, where $T_i$ and $T_j$ are two wall-clock timestamps and $T_i <
T_{j}$.
The length of the period $T_j-T_i$ is a tunable parameter, which is called the
\textit{statistics period length} (SPL). A load value of a particular resource
is a percentage point in the range $[0,100]$.   

Based on the statistics, we detect the \emph{bottleneck resource} of the
computation, i.e. the one with the greatest total usage in the whole system.
The load balancing objective will use the load values of the bottleneck
resource. We define $gLoad_k$ and $load_i$ as the average load value of the
bottleneck resource over the latest $SPL$ of a key group $g_k$ and a node $n_i$
respectively.

\textbf{Workload Fluctuations.}
In this paper, we distinguish between short-term and long-term workload
fluctuations. There already exist techniques to handle short-term fluctuations,
such as data buffering, back-pressure and frequent minor adaptations
like~\cite{shah:2003:flux}. But these techniques still cannot solve all the
problems especially with very high load spikes. For example, back-pressure and
data buffering may build up a very long queue or even spill data to disks and
hence may incur excessive latency; and frequent adaptation may incur high
communication overhead and reduces the system throughput. 

As shown by the analysis in~\cite{Xing:2006:PRL:1182635.1164194}, a more
balanced long-term load distribution can be more resilient to short-term
fluctuations by reducing the probability and the degree of short-term
overloading at any particular node and can alleviate the aforementioned problems
caused by short-term spikes. This paper does not aim to replace the techniques
for handling short-term spikes, but rather focuses on ensuring a good load
balance over a relatively long period of time.


\textbf{Heterogeneity.}
It is not assumed the nodes in the cluster are homogeneous, so in order to
compare two load values, they must be multiplied with the node capacity. The
constants can be inferred at runtime, by first assuming all nodes are
homogeneous, then measuring the actual effect of state migration. Notice that
heterogeneous performance cannot be expected even for nodes of the same type,
due to their locations in a data center, network considerations and other
factors outside the control of the system.

\textbf{Controller.} The controller is a system level operator that makes global
decisions. The controller is responsible for collecting statistics from all
operator instances and making them easily accessible to the adaptation
algorithm. Furthermore, it also runs the adaptation algorithms periodically.


%


\section{Integrative Reconfiguration}

\subsection{Integrative Adaptation Framework}
The adaptation framework takes horizontal scaling, load balancing and operator
instance collocation into account. Here horizontal scaling is the ability to
dynamically scale the number of nodes in the cluster, while load balancing and
operator instance collocation are, respectively, to balance the workload over a given set of
nodes and to minimize the cost of data communication by allocating the key
groups of the operators.

\begin{algorithm}[t]
\DontPrintSemicolon 

\For{each node $n_i \in N$ that is set to be removed by the scaling algorithm in previous periods} {
	\If{$n_i.keygroups$ is \textbf{empty}} {  
		\textbf{terminate} node $n_i$\\ 
	}
}
plan $\leftarrow$ keyGroupAlloc()  \Comment{the allocation plan}\\
	\If {Scaling(plan)} {
		// Wait until new nodes are allocated\\
		plan $\leftarrow$ keyGroupAlloc() \Comment{recalc after scaling}\\
}
apply(plan) \\
\caption{Adaptation Framework}
\label{alg:adaptation_overview}
\end{algorithm}

Algorithm~\ref{alg:adaptation_overview} presents our simple yet effective
integrative adaptation framework. The adaptation is run periodically. In lines
1-3, it starts by checking all the processing nodes that are marked for removal
in the previous adaptation periods, and if all the key groups have been moved
out from these nodes, then they can be safely removed from the job.

After that, the algorithm calculates a potential allocation plan for the
operator instances (line 4), which factors in both load balancing and collocation
of operator instances. This new allocation plan will not be deployed
immediately, but will be used when making decisions about whether horizontal
scaling is needed (line 5). This is important to avoid unnecessary or
undesirable scaling because:
\begin{compactitem}
  \item A potential allocation optimizes load balancing that may solve the
	overloading problem of a processing node without scaling;
  \item Collocation of communicating operator instances may decrease the total
	load on the system so that scaling out can be avoided; 
  \item (Undesirable) Scaling-in will not be done if it is impossible to balance
	the load well enough among the remaining nodes. 
\end{compactitem}

Furthermore, after making a scaling decision in line 5, we will redo the allocation
planning algorithm again (line 7), which will make an integrative decision for
scaling, balancing and collocation. As discussed in the following subsections,
we put a constraint on the overhead of state migration within each adaptation
period. Therefore, the allocation algorithm needs to decide which operator
instances should be migrated to solve the more urgent problems. For example, it
may decide to migrate the load away from an overloaded node instead of moving
the load away from a node marked for removal by the scaling algorithm.

\subsection{Horizontal Scaling}
The scaling decision~\cite{ssdbm/MadsenTZ14,
ElasticScalingForDataStreamProcessing, DBLP:journals/corr/FuDMWYZ15} is 
determining how many nodes are needed to process the current workload. Previously
proposed algorithms, such as~\cite{ElasticScalingForDataStreamProcessing,
DBLP:journals/corr/FuDMWYZ15}, can be directly applied in our framework. As
developing a novel scaling optimizer is outside the scope of this paper and the
actual algorithm would not affect the conclusion of this paper, we
assume the use of existing techniques for calculating the number of needed nodes.
In addition, we do not assume a particular way to add new nodes to a job, which
can be done by starting new instances in Amazon EC2, waking up nodes that were put into
sleep for energy saving, or reallocating resources from other jobs to this one.

\subsection{Key Group Allocation}
\label{subsec:rebalance}



\subsubsection{LP Solver}
\label{subsubsec:lp_solver}
The following solution can be applied to situations where there is little
opportunity to minimize communication cost by collocating key groups. To
analyze different situations, we consider four common partitioning
patterns,
which are similar to those proposed in~\cite{Zhou:2012:APT:2213836.2213839}:

\begin{figure}[tbh] \centering
\includegraphics[width=235px]{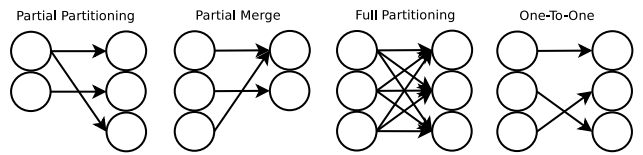}
\caption{Common Partitioning Patterns}
\label{fig:partitioning_situations}
\end{figure}

\begin{compactitem}

  \item Partial Merge: Each operator instance outputs all tuples to one
  downstream operator instance. 
    
  \item Partial Partitioning: Each operator instance outputs tuples to a
  subset of downstream operator instances.

  \item One-To-One: Each operator instance outputs tuples to exactly one target
  operator instance, and each target instance receives tuples only from one
  upstream operator instance. 

  \item Full Partitioning: Each operator instance outputs tuples to all
  downstream operator instances. This is a special case of Partial Partitioning.


\end{compactitem}

\begin{table}[thb]
\center
\begin{tabular}{l p{6.4cm}}
\bf{Symbol} & \bf{Description} \\
\hline
$A$ & Nodes are not marked for removal\\
$B$ & Nodes are marked for removal by the scaling algorithm\\
$G$ & All the key groups\\
$q_{i,k}$ & binary variable indicating if $g_k$ is currently allocated to $n_i$\\
$x_{i,k}$ & binary variable indicating if $g_k$ is allocated to $n_i$ in the
new solution \\
$mean$ & $\lceil \frac{1}{|A|} \cdot \sum_{n_i \in N} load_i \rceil$, i.e. the
average load\\
$d$ & Maximum load deviation from $mean$ $\forall n_i \in A$\\
$d-d_u$ & Maximum upper load deviation from $mean$\\
$d-d_l$ & Maximum lower load deviation from $mean$\\
$w_1,w_2$ & Weights in the objective function\\
$mc_k$ & Migration cost for $g_k$\\
\end{tabular}
\caption{MILP Symbols}
\label{tab:symboloverviewMILP}
\end{table}

The LP Solver described in this section, is thus suitable for topologies
exhibiting the two partial patterns with high degrees (including Full
Partitioning patterns) when the amount of data transmitted from an instance to
different downstream instances are not skewed. In such cases, there is limited
benefit of collocating key groups, because data are evenly sent to many
downstream instances.


%

{\bf Metric.} To measure load imbalance, we define a metric called
\textit{load distance}, which is the largest difference between the exhibited
load of any node in the cluster and the average load in the cluster. We would
like to find a load balancing solution that minimizes the load distance, because
this is the state where the system can tolerate maximum load fluctuations at any
node, without exhibiting under- or overload. 

Let $B$ be the set of nodes that are marked for removal by the horizontal
scaling algorithm and $A$ be the rest of the processing nodes such that $N = A
\cup B$.  Then the average load, $mean$, is defined as $\lceil \frac{1}{|A|} \cdot
\sum_{n_i \in N} load_i \rceil$, where $load_{i}$ is the load of $n_i$. The
objective can be formally stated below:

\textbf{Objective:} \textit{Minimize $\max_{n_i \in A}$ $|load_{i} - mean|$ and
\\ $\sum_{n_i \in B} (load_{i})$ s.t. the cost of migration $\leq$ maxMigrCost.}

It is important to bound the cost of rebalancing, as the solution could
otherwise end up being very costly to apply within each adaptation round and
hence incur excessive processing latency. We model the cost of migrating a
keygroup $g_k$ as $mc_k = \alpha \cdot |\sigma_k|$, where $|\sigma_k|$ is the
size of the state of keygroup $g_k$ and $\alpha$ is a constant, chosen such that
$mc_k$ is the time to serialize the state on a node with average load. Our
techniques are largely independent on the cost-model chosen, and 
the cost-model can be chosen according to the actually employed state migration
technique. As the maximum migration cost is bounded, horizontal scale-in will
not necessarily be done within one adaptation round, but instead will gradually
migrate key groups from $B$ to $A$.

The minimization problem is NP-hard, because it is an instance of the
\textit{Multi-Resource Generalized Assignment Problem} \cite{NAV:NAV1029}. To
derive a solution, we model the minimization problem as a \textit{Mixed-Integer
Linear Program} as follows.

\textbf{Variables.}
Table~\ref{tab:symboloverviewMILP} summarizes the symbols. The current
allocation of key groups to nodes are represented by a set of binary variables
$q_{i,k}$ with a value of $1$ if key group $g_k$ is located at node $n_i$ or $0$
otherwise. The new allocation of key groups is denoted by the binary variables
$x_{i,k}$, which is defined in a similar way as $q_{i,k}$. A node $n_i$ is
marked for removal if $kill_i=1$. 
The variables $d$, $d_u$ and $d_l$ model the load deviations from $mean$. 



\textbf{Mixed-Integer Linear Program}
\begin{equation*}
\begin{aligned} & \text{min} & & w_1\cdot d - w_2(d_u + d_l)\\ 
& \textit{s.t.} & &  \\
& (1) & & \forall_{g_k \in G} : \sum_{n_i \in N} x_{i,k} = 1 \\
& (2) & & \sum_{n_i \in N} \sum_{g_k \in G} [(1-q_{i,k}) \cdot x_{i,k} \cdot mc_k] \leq maxMigrCost \\
& (3) & & \forall_{n_i \in N} : \sum_{g_k \in G} (x_{i,k} \cdot
gLoad_k) \leq mean + (d - d_u) \\
& (4) & & \forall_{n_i \in N \wedge kill_i=0} : \sum_{g_k \in G} (x_{i,k} \cdot gLoad_k) \geq mean - (d - d_l) \\
& (5) & & mean - d \ge 0\\
\end{aligned}
\end{equation*}

\textbf{Constraints.} Constraint (1) ensures that each key group is allocated to
exactly one node. Constraint (2) bounds the maximum cost of state migration.
Constraint (3) bounds the maximum load per node. Constraint (4) bounds the
minimum load per node, and is effectively disabled for nodes marked for removal.
Constraint (5) is needed to ensure that $d$ does not exceed the lower bound of
load.

\textbf{Explanation.}
The MILP minimizes the load distance, with the help of three variables and two
constraints. Constraints (3) and (4) define the maximum and minimum load of
each node in the cluster, where the limits are influenced by the variables $d$,
$d_u$ and $d_l$, such that the difference between maximum and minimum load can
be minimized by choosing appropriate values of the variables.

The variable $d$ represents the maximum load deviation from $mean$ in all the
nodes in $A$. By minimizing the value of $d$, it is guaranteed that at least the
upper or the lower bound of load on all nodes will be tight. In order to make
both the upper and lower bounds tight, we introduce the variables $d_u, d_l \in
\mathbb{R}$, such that $d_u$ (or $d_l$) is chosen to be greater than zero to
tighten the upper (or lower) bound.

The solution thus depends on the variable $d$ being minimized first and then
secondarily $d_u$ and $d_l$ being maximized. To achieve this, the constants
$w_1$ and $w_2$ in the objective function should be
chosen such that $w_1 >> w_2$.

\textbf{Extending to Multi-Dimensional Load.}
For ease of presentation, the above formulation only considers 1-dimen-sional
load value, i.e.  the usage of the bottleneck resource. This may be sufficient
for many computations where the usage of different resources are correlated,
e.g.  higher data inputs usually mean higher usage of CPU (due to more
computations and more data serialization and deserialization), memory and
network bandwidth.  Therefore, balancing the usage of the bottleneck resource
may also bring a good distribution of the usage of other resources. If this is
not the case, then thanks to the flexibility of the MILP model, it can be easily
extended to add a constraint on the maximum usage of each non-bottleneck
resource on each node.

\textbf{Extending to Heterogeneous Nodes.}
Heterogeneity is supported by multiplying the $gLoad_k$ constant (constraints 3
and 4) by a node capacity weight, thus forming a new constant. Changing the
values of constants will change the corner points of the solution and not the
performance/quality (e.g. using Simplex).

\textbf{Supporting Horizontal Scale-In.}
When the horizontal scaling algorithm decides to scale-in, it marks a set of
nodes, $B$, for removal and the load balance algorithm should migrate key groups
away from the nodes to be removed when appropriate.
A careful reader may notice that our MILP will primarily minimize load distance
$d$, and we do not explicitly give higher priority to the migrations of key
groups from $B$ to $A$. Here we will prove that solving our MILP will eventually
move all the key groups from $B$ to $A$.


%
\newtheorem{lemma}{Lemma}

\begin{lemma}No key group will be migrated from $A$ to $B$.\end{lemma}

\begin{proof}
We first define how to choose four nodes $n_1$ to $n_4$ $\in A$ from an
arbitrary cluster of $N$ nodes. Without loss of generality, we assume $load_1 > 
load_2 > load_3> load_4$ and $\forall_{i \not \in\{1,2,3,4\}}
load_2>load_i>load_3$. The load of each node is modelled as follows: $load_1 =
mean + h$, $load_2 = mean + i$, $load_3 = mean - k$ and $load_4 = mean - j$.
Lastly, the key group $g_t$ to move is located on node $n_1$ and key group $g_t$
has load $gLoad_t$. The proof proceeds by showing that migrating $g_t$ to or
within $A$ is always preferred.

Case 1, $j \ge h$:
If migrating $g_t$ to $B$, the resulting value of $d = j$, remains unchanged. In
other words, this operation cannot decrease the value of $d$.
Consider now migrating $g_t$ within $A$. The value of $d$ is decreased by
min($gLoad_k$, $j-k$, $h-i$), which is always greater than zero. To see why $d$
is decreased by this amount, see that the statement actually considers three
aspects: (1) size of migrated load, (2) maximum possible reduction in underload
and (3) maximum possible reduction in overload. Taking the minimum of these
three, actually gives the bounding one, which is how $d$ is defined.

Case 2, $j < h$:
If migrating $g_t$ to $B$, the value of $d$ is decreased by min($gLoad_k$,
$h-j$, $h-i$). By migrating $g_t$ within $A$, $d$ is instead decreased by
min($gLoad_k$, $j-k$, $h-i$). Both calculations considers three aspects: (1)
size of migrated load, (2) maximum possible reduction in underload and (3)
maximum possible reduction in overload, as before.

First we argue that $j-k > h-j$, because this proves that $d$ is decreased by at
least as much when moving $g_t$ within $A$, compared to moving $g_t$ to $B$.
$j-k > h-j$ $\leftrightarrow$ $h-k > h-j$ $\leftrightarrow$ $-k > -j$
$\leftrightarrow$ $k < j$, therefore the statement holds.

As $d$ can be decreased by the same amount by moving $g_t$ within $A$ or to $B$,
we now need to prove that migrating within $A$ is still the preferred choice in
this case. First, see that the variable $d_u$ (as defined in MILP) is changed by
the same value independent of whether $g_t$ is moved within $A$ or to $B$, since
the reduction in overload is the same. Secondarily, consider if migrating $g_t$
to $B$, the value of $d_l$ (as defined in MILP) is unchanged as no underloaded
node in $A$ receives any load. Lastly, see that if migrating $g_t$ within $A$,
the value of $d_l$ is increased by min($gLoad_k$, $j-k$). Since $d_u + d_l$ is
therefore increased by moving $g_t$ within $A$, and since $d_u$ and $d_l$ are
secondarily maximized by the MILP, it is preferred to move $g_t$ within $A$.

We now consider $|A| < 4$. The case when $A$ contains only one node is trivial,
as $g_t$ will only be moved to $B$ when the node is overloaded, which by
definition is impossible. If $|A| = 3$, then either $i = 0$ or $k = 0$ and when
$|A|=2$, then $i=k=0$. By using these conditions, one can easily prove that the
lemma holds for both cases. To keep the proof succint, we ommit the details.
\end{proof}

\begin{lemma}
  The minimum value of $d$ can only be achieved by migrating all key groups
from nodes in $B$ to nodes in $A$.
\end{lemma}

\begin{proof}
We denote the sum of loads on the nodes in $A$ and $B$ as $L_A$ and $L_B$,
respectively, and $L_N=L_A + L_B$. Consider not moving any key group from nodes
in $B$ to those in $A$. The value of $d \ge \frac{L_N}{|A|} - \frac{L_A}{|A|} =
\frac{L_B}{|A|}$, which follows directly from the definition of $d$ (constraints
3 and 4). Now consider allowing the migration of key groups from nodes in $B$ to
nodes in $A$. The value of $d \ge \frac{L_N}{|A|} - \frac{L_A}{|A|} -
\frac{L_B}{|A|} = 0$. This again follows directly from the definition of $d$.
\end{proof}

\subsubsection{Autonomic Load Balancing with Integrated Collocation}
Data communication between key groups on separate processing nodes consumes not
only bandwidth, but also CPU, as it requires data serialization and
deserialization. Collocating key groups that does a significant amount of
communication can therefore save both bandwidth and CPU, and thus reduce the
system load. 

Topologies exhibiting extensive One-To-One partitioning patterns, 
the two partial patterns with low-degrees, or even the two partial patterns with
high-degrees (including full partitioning) but high skewness, could have
abundant opportunities to minimize communications by collocating key groups.
Note that we assume the computation of each operator is opaque to the system and
we cannot deduce the relations between the input key and output key of an
operator. Therefore, we cannot perform a pre-analysis of the communication
patterns in the topology and produce an optimized collocate plan statically as
done in StreamCloud~\cite{Gulisano:2012:SES:2412378.2412792}. In other words, we
have to dynamically detect the communication pattern and its changes over time,
and then optimize the collocation plan dynamically. 


The above MILP cannot be simply extended to model the collocation of key groups,
as detecting if two key groups are collocated needs a quadratic formulation. For
efficiency reasons, we avoid quadratic models, as they are computationally much
more expensive to solve than linear models. Therefore we propose a heuristic
algorithm to run on top of the MILP, called ALBIC (Autonomic Load Balancing with
Integrated Collocation).  ALBIC will dynamically measure the benefit of
collocating each pair of key groups. For each invocation, ALBIC optimistically
collocates one set of key groups with maximum benefit. During dynamic load
balancing, the collocated key groups would be considered for migration as
indivisible units.  If the load of a set of collocated key groups becomes too
large, it will be split into relatively even-sized partitions. In this way,
ALBIC optimizes the key group collocation integratively with horizontal scaling
and load balancing.


%

\begin{algorithm}
\DontPrintSemicolon 
\KwIn{maxLD (max load distance, default = 10) ~~~~~~~~~~~~~ maxPL (max partition load,
initial = 25) ~~~~~~~~~~~~~ stepPL (change in partition load, default = 5) 
~~~~~~~~~~~~~~~~~~~ sF (score factor, default = 1.5)}
\textbf{// Calculate scores (step 1)}\\
\For {each operator $O$} {
	\For {each keygroup $g_k$ $\in$ O.keygroups} {
		output $\leftarrow \sum_{DO \in O.down} \sum_{g_j \in DO.kgs} out[g_k][g_j]$\\
		avg $\leftarrow output / \sum_{DO \in O.down} |DO.keygroups|$\\
		\For {each operator $DO$ $\in$ O.downstream} {
			\For {each key group $g_j$ $\in$ DO.keygroups} {
				\If {out[$g_k$][$g_j$] $>$ avg $\cdot sF$} {
					\If {$g_k$, $g_j$ are collocated} {
						add $g_k, g_j$ to \textit{colGrps}\\
					} \Else {
						add $g_k, g_j$ to \textit{toBeColGrps}\\
					}
				} 
			}
		}
	}
} 

\textbf{// Maintain collocation (step 2)}\\
sets $\leftarrow$ calcSets(colGrps)\\
\For {each set $S \in$ sets} {
	$p_1 \leftarrow$ $\lceil \sum_{g_k \in S} (g_k.migrCost) / maxMigrCost \rceil$\\
	$p_2 \leftarrow$ $\lceil \sum_{g_k \in S} (gLoad_k) / maxPL \rceil$\\
	 \For {each set $P \in$ graphPart(S, max($p_1,p_2$))} {
	// contraint: $P$ is migrated as a unit\\
	partitions $\leftarrow$ partitions + P\\
	}
}

\textbf{// Improve collocation (step 3)}\\
$g_i, g_j \leftarrow$ random from \textit{uncolGrps} w. max value\\
$n_1 \leftarrow$ keygrpToNode($g_i$), $n_2 \leftarrow$ keygrpToNode($g_j$)\\
\If {$g_i, g_j \notin$ partitions} {
	Add constraint: $g_i,g_j$ on node w. load=$min(l_1,l_2)$
} \ElseIf {$g_i \in$ partitions \textbf{AND} $g_j \notin$ partitions} {
	Add constraint: $g_i,g_j$ on node $n_1$\\
} \ElseIf {$g_i \notin$ partitions \textbf{AND} $g_j \in$ partitions} {
	Add constraint: $g_i,g_j$ on node $n_2$\\
} \ElseIf {$g_i, g_j \in$ partitions} {
	$p_1 \leftarrow$ getPartition($g_i$), $p_2 \leftarrow$ getPartition($g_j$)\\
	Add constraint: $p_1,p_2$ on node w. load=$min(l_1,l_2)$
}

\textbf{// Solve (step 4)}\\
solution $\leftarrow$ lp-solver(this)\\
\If {calcLoadDistance(solution) $>$ maxLD} {
	\textbf{return} albic(maxLD, maxPL-stepPL, stepPL, sF)\\
}
\textbf{return} solution\\

\caption{ALBIC}
\label{alg:albic}
\end{algorithm}

\textbf{Step 1 - Calculate Scores.}
Let $out(g_i)$ be the total data rate sent from a key group $g_i$ within
the timespan SPL and let $out(g_i, g_j)$ be data rate sent from key
group $g_i$ to $g_j$. 
To decide if a keygroup pair $g_i, g_j$ can contribute to the overall
collocation, the value of $out(g_i,g_j)$ should exceed a threshold value $avg(g_i)
\cdot sF$, where the $avg(g_i)$ is defined as the total number of tuples sent downstream from
$g_i$ divided by total number of downstream key groups and $sF$ is a score
factor. Setting $sF=1$ means key groups sending more than average to each other
will be considered. Setting $sF=2$ means key groups sending more than twice the
average to each other will be considered, and so on.


\textbf{Step 2 - Maintain Collocation.}
ALBIC merges all existing collocated key group pairs into a minimum number of
sets, such that any pair of sets whose intersection is not empty, will be
replaced by the union of those two sets. A set cannot be too large, as that
hinders good load balancing and state migration. To overcome this problem, each
set is split into a number of partitions, attempting to ensure that (1) the
migration cost of a partition $Pmc$ is less than \textit{maxMigrCost} and (2)
the maximum load $PL$ of any partition is less than $maxPL$. ALBIC uses a graph
model, where each key group is modelled as a vertex and each edge has $weight =
out(g_i,g_j)$. The weight of the vertex of $g_i$ is set as the $mc_i$, its
migration cost, if $Pmc/maxMigrCost > PL/maxPL$. Otherwise it is set as
$gLoad_i$, the load of $g_i$. Ties are broken randomly.
Balanced graph partitioning \cite{Karypis:1998:MAM:509058.509086} is then
applied to generate balanced partitions with minimum inter-partition weighted
edge-cuts. It may need to be applied again on some partitions if they still
violate one of the two constraints.

\begin{table}[t]
\center
\begin{tabular}{p{1.2cm} p{6.4cm}}
\bf{Symbol} & \bf{Description} \\
\hline
$out(g_i)$ & The output rate of $g_i$\\
$avg(g_i)$ & $out(g_i)$ divided by \#downstream key groups \\
$out(g_i,g_j)$ & The rate of data sent from $g_i$ to $g_j$\\
$Pmc$ & Migration cost for unspecified partition\\
$maxLD$ & Maximum load distance\\
$maxPL$ & Maximum load of any partition\\
$stepPL$ & Decrease in $maxPL$ when recalculating\\
$sF$ & Score factor\\
\end{tabular}
\caption{ALBIC Symbols}
\label{tab:symboloverviewALBIC}
\end{table}



\textbf{Step 3 - Improve Collocation.}
Select one random key group pair $g_i, g_j$ from the set
\textit{toBeColGrps} with the maximum value of $out(g_i,g_j)$. Let $n_1$
and $n_2$ be the node where $g_i$ and $g_j$ are currently allocated
respectively. There are three cases:

\textit{Case 1: neither $g_i$ or $g_j$ is part of any set of collocated
key groups}. This case is handled by adding a constraint to the MILP to
collocate $g_i, g_j$ on the node $n_1$ or $n_2$ with the smallest load.

\textit{Case 2: either $g_i$ or $g_j$ is part of a partition of collocated
key groups}. This case is handled by adding a constraint to the MILP to collocate
$g_i, g_j$ on the node, where the key group that is part of a set of
collocated key groups is located.
 
\textit{Case 3: $g_i$ and $g_j$ belong respective sets of collocated
key groups}. This case is handled similarly as Case 1.

ALBIC always tries to collocate one pair of key groups with the largest amount
of communication to reduce the cost of serialization and deserialization. 

\textbf{Step 4 - Solving.}
The constrained MILP is solved and the load distance of the resulting allocation
is calculated. If the load distance is larger than $maxLD$ (the user-defined
maximum load distance), the problem is resolved by forming smaller (more)
partitions by reducing the $maxPL$ (max partition load) parameter.
Notice that when $maxPL = 0$ there must be exactly one partition per key group,
in which case ALBIC simply solves the pure MILP, without considering collocation
at all.

In the last part of this description, we provide a discussion of the default
values of arguments to ALBIC. As there is a trade-off between the value of
$maxLD$ and the collocation which can be obtained, we use a default value of
$maxLD=10$, which guarantees a reasonable load distance without impacting the
obtainable collocation unnecessary. ALBIC achieves a much lower load distance
than $maxLD$ for all experiments conducted. The initial value of $maxPL$ is
25\%, which makes it very unlikely this initial value impacts the obtainable
collocation while ensuring no partition gets very large load, for load balancing
purposes. If the constraint of $maxLD$ is violated, the value of $maxPL$ is
gradually decreased by $stepPL$ and ALBIC will therefore use more and more
partitions to produce a solution respecting the constraint of $maxLD$.
The default value of $stepPL=5$ is chosen such that the solution will not need
to be recalculated too many times (max five times). In our experiments, 
it is very rare that any recalculation of ALBIC is needed.

\enlargethispage{\baselineskip}
\section{Experiments}

\begin{figure*}[h!t]
\centering
\begin{minipage}{4.38cm}
	\centering
	\includegraphics[width=110px]{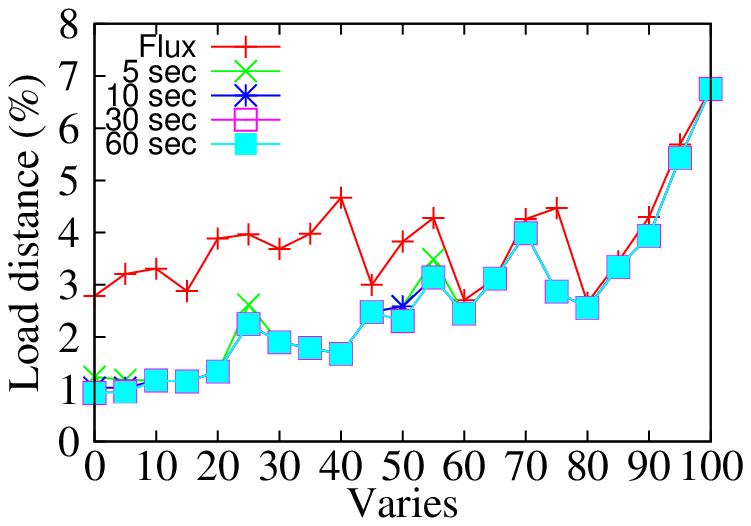}
	MaxMigrations = 10
\end{minipage}%
\begin{minipage}{4.38cm}
	\centering
	\includegraphics[width=110px]{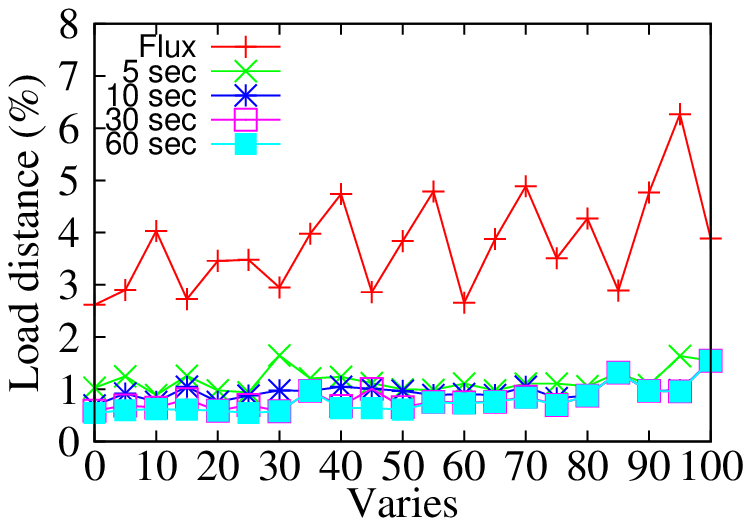}
	MaxMigrations = 20
\end{minipage}
\begin{minipage}{4.38cm}
	\centering
	\includegraphics[width=110px]{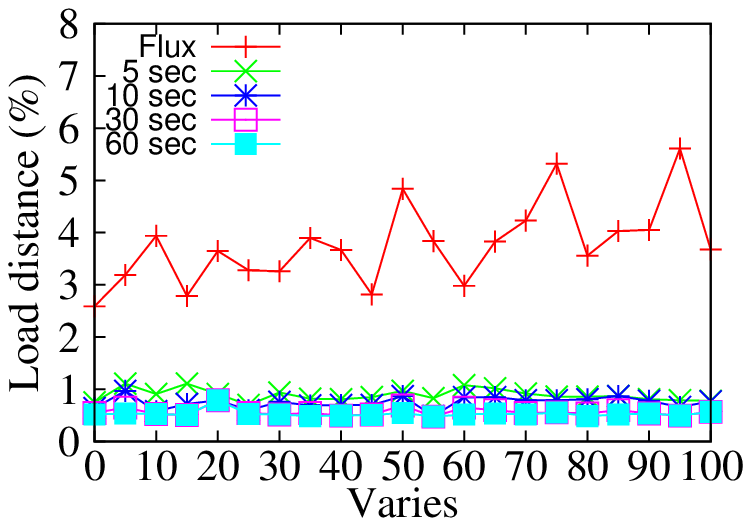}
	MaxMigrations = 30
\end{minipage}
\begin{minipage}{4.38cm}
	\centering
	\includegraphics[width=110px]{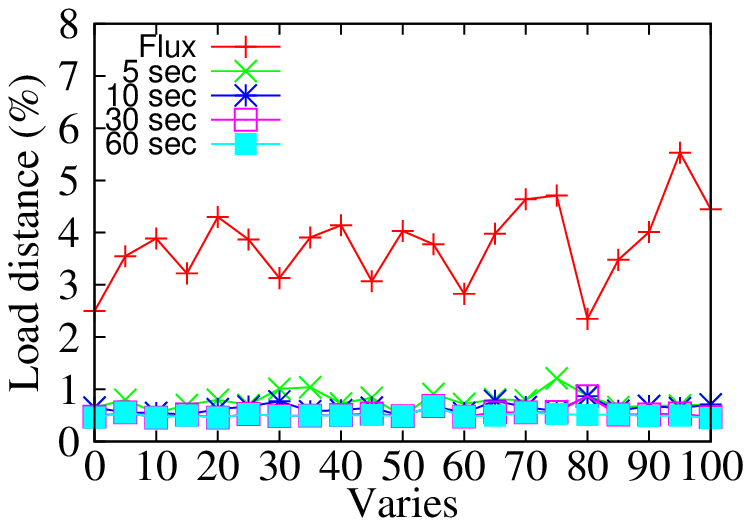}
	MaxMigrations = 40
\end{minipage}
\caption{Experimental Setting: 20 nodes, 400 key groups, 10 operators}
\label{eval:milp_performance_20nodes}
\end{figure*}

\begin{figure*}[h!t]
\centering
\begin{minipage}{4.38cm}
	\centering
	\includegraphics[width=110px]{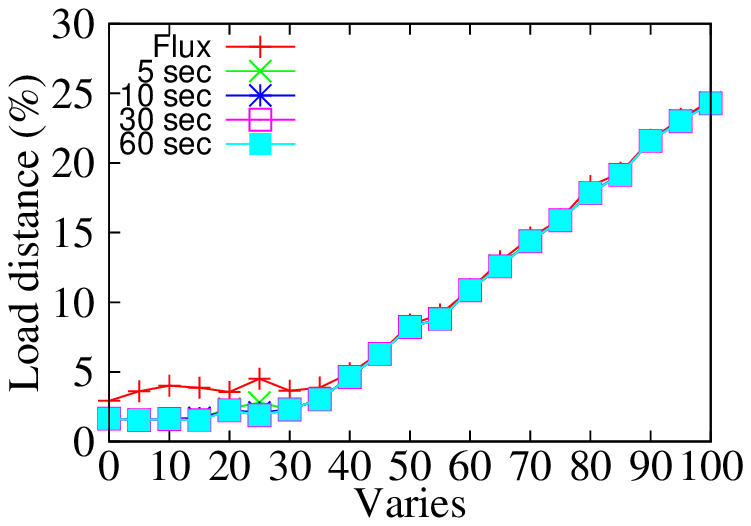}
	MaxMigrations = 10
\end{minipage}%
\begin{minipage}{4.38cm}
	\centering
	\includegraphics[width=110px]{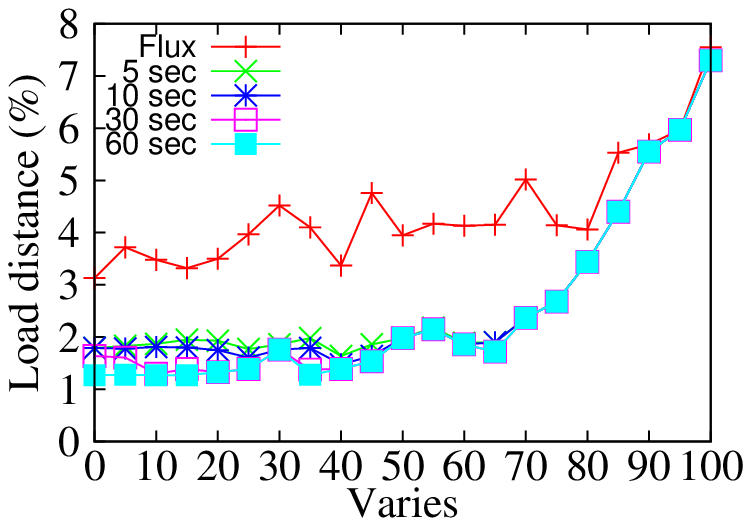}
	MaxMigrations = 20
\end{minipage}
\begin{minipage}{4.38cm}
	\centering
	\includegraphics[width=110px]{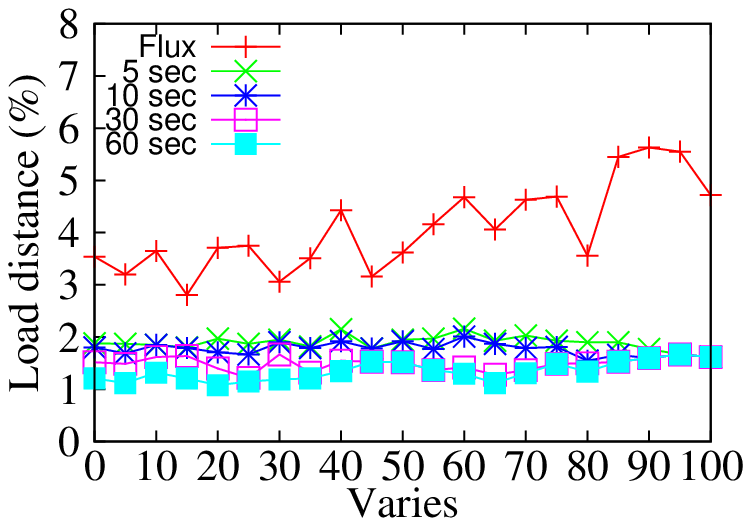}
	MaxMigrations = 30
\end{minipage}
\begin{minipage}{4.38cm}
	\centering
	\includegraphics[width=110px]{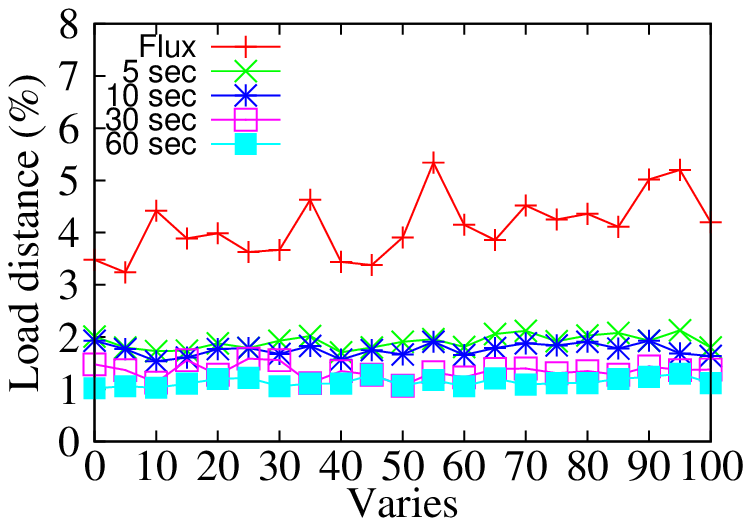}
	MaxMigrations = 40
\end{minipage}
\caption{Experimental Setting: 40 nodes, 800 key groups, 20 operators}
\label{eval:milp_performance_40nodes}
\end{figure*}

\begin{figure*}[h!t]
\centering
\begin{minipage}{4.38cm}
	\centering
	\includegraphics[width=110px]{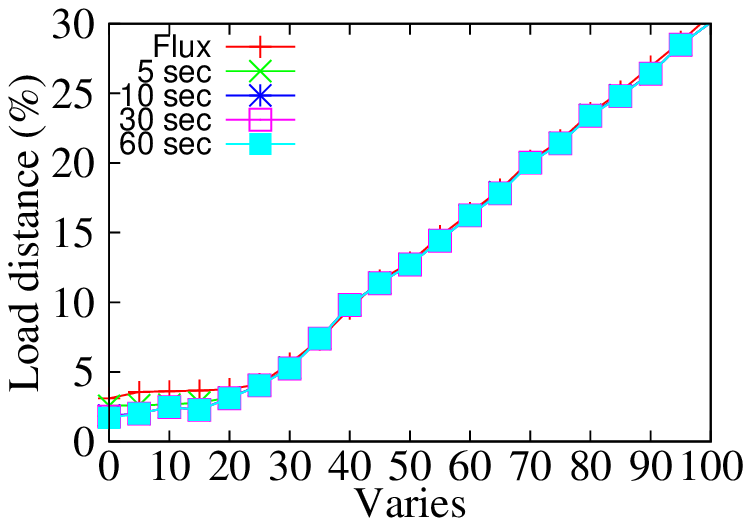}
	MaxMigrations = 10
\end{minipage}%
\begin{minipage}{4.38cm}
	\centering
	\includegraphics[width=110px]{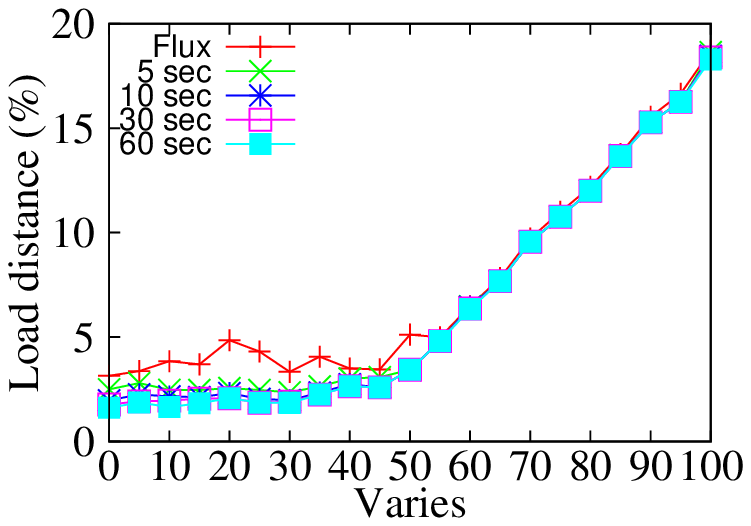}
	MaxMigrations = 20
\end{minipage}
\begin{minipage}{4.38cm}
	\centering
	\includegraphics[width=110px]{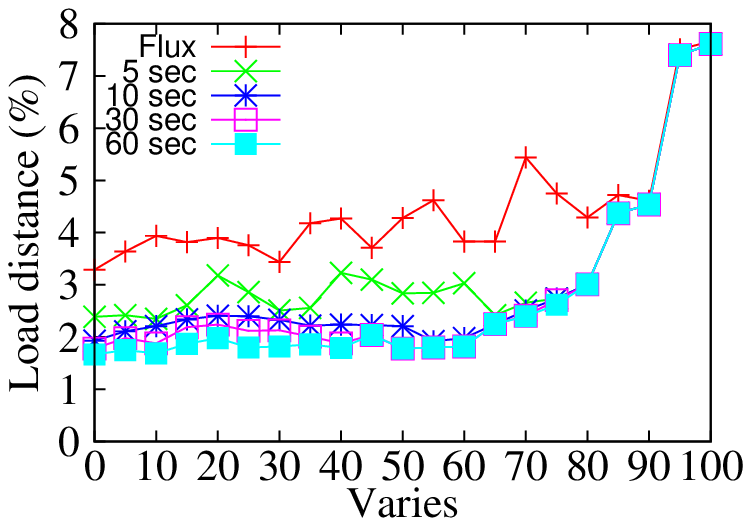}
	MaxMigrations = 30
\end{minipage}
\begin{minipage}{4.38cm}
	\centering
	\includegraphics[width=110px]{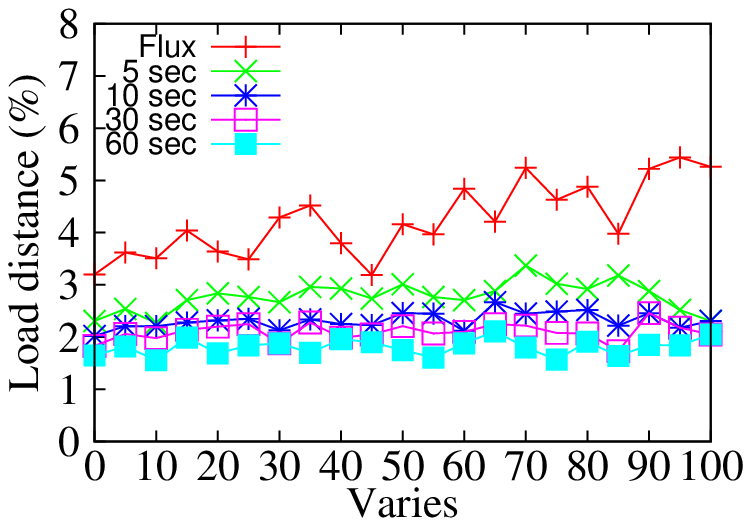}
	MaxMigrations = 40
\end{minipage}
\caption{Experimental Setting: 60 nodes, 1200 key groups, 30 operators}
\label{eval:milp_performance_60nodes}
\end{figure*}

\textbf{Hardwares.}
All our distributed experiments are conducted on Amazon Elastic Compute Cloud
(EC2) with Apache Storm. A cluster consists of one \textit{master node}, between
one and four \textit{input nodes} and between five and twenty \textit{worker
nodes}. The master node has instance type $m1.medium$, and executes the Storm
master, Apache Zookeeper \cite{Hunt:2010:ZWC:1855840.1855851} and the Storm UI.
The worker nodes have instance type $m1.medium$ and are used to execute the
actual job logic. The input nodes have instance type $m3.xlarge$ and are used to
produce input to the processing nodes.
A few optimizer experiments are also executed \textit{locally}, on a single
desktop computer, with a Core i7-2600K (3.4Ghz) processor and 8G of memory,
running Windows 8.1.


\textbf{Initialization.}
As Apache Storm is based on Java, it requires an initialization phase after
deployment, where the Just-In-Time compiler can do runtime optimizations. To
respect this, we ignore the unstable initialization phase for each
experiment, which can be seen from the missing initial time periods in our
figures.

\textbf{Metrics.}
We use \emph{load distance} to measure the load imbalance. As shown in previous
work~\cite{Xing:2006:PRL:1182635.1164194, Lei:2014:RDQ:2627748.2602138}, a more
balanced load distribution can be more resilient to short-term load fluctuations
by minimizing the chances of overloading and the queueing latency during load
spikes.  

To verify the collocation of communicating key groups can reduce the system
workload, we use the metric \emph{load index} to measure how the system load is
changed over time. It is defined as the current average system load 
divided by the average system load right after the initialization phase. 

\textbf{Techniques.}
The MILP (and ALBIC) is implemented using \textit{IBM ILOG CPLEX} v$12.6.1$
\cite{CplexWebsite} and \textit{Metis} v$5.1$ \cite{MetisWebsite}. 
In addition, all the approaches are executed with default parameters stated in
the previous sections. 

We compare our MILP with Flux~\cite{shah:2003:flux} and PoTC~\cite{7113279},
which can achieve dynamic load balancing. As Flux controls overhead by limiting
the number of state migrations, we modify our MILP in the same way (i.e. we
change from a limit on $maxMigrCost$ to $maxMigrations$).  
We also compare ALBIC with COLA~\cite{khandekar2009cola}, which, to the best of
our knowledge, is the only approach that achieve similar objectives as ALBIC.
As COLA is a static optimizer and does not consider runtime adaptations,
invoking it for each adaptation period would incur massive state migrations. The
comparison is to show how well ALBIC adapts a key group allocation
plan in comparison to a complete re-optimization. 
%

\textbf{Datasets.} 
We use three datasets for our experiments.  The dataset \textit{Parsed Wikipedia
edit history} \cite{snapnets}, contains the complete Wikipedia edit history up
to January 2008. The data is rich with minimum $14$ attributes for each of the
$116,590,856$ article revisions.  The data input rate is fluctuating in the
order of hundreds of tuples per second. To better illustrate the capability of
our solution, we scale the size of the data, while maintaining the input
distribution.

The dataset \textit{Airline On-Time} \cite{BtsWebsite} is provided by the
Research and Innovative Technology Administration, United States Department of
Transportation. We use data from the period January 2004 to December 2013, which
contains information about airplanes, such as departure, arrival, expected time
of arrival.

The dataset \textit{Global Surface Summary of the Day} \cite{GsodWebsite} is
provided by the National Oceanic and Atmospheric Administration, USA. We use
data from the period January 2004 to December 2013. The data contains mean
temperature, mean visibility, precipitation and more, for each of the several
thousand weather stations.

\subsection{Solver Performance for MILP}
\label{subsec:milp_solver_performance}
We first investigate the relationship between solving time and
solution quality for the MILP. The
experiment was executed with the \textit{local} desktop.

As the performance of the solver is dependent on multiple factors such as
cluster size and load distribution, we use synthetic data to simulate these
scenarios. Key groups are evenly allocated, such that each node has the same
number of key groups. The load of each key group is initially set to the mean,
which is then adjusted by a percentage, randomly chosen from the range
$[-5\%, 5\%]$. This is to simulate that in a realistic setting, the nodes that
process these key groups have load fluctuation from one to another.  Now we
further adjust the load on $20\%$ of the nodes, which is controlled by a
variable called \textit{varies}. Half of the nodes whose load are changed gets a
reduction of $0.5 \cdot varies$ in their load and the other half gets an
increase in load of $0.5 \cdot varies$. The load changes are done by modifying
the load of a randomly selected set of key groups on a node.

Figures~\ref{eval:milp_performance_20nodes} -
\ref{eval:milp_performance_60nodes} show the results of three different
clusters, with varying values of maxMigrations and load fluctuations.
The results show that our MILP approach consistently outperforms Flux, even
after just a few seconds of optimization by the solver. The MILP approach
quickly converges towards a pretty good solution (within a few seconds), which
can be improved only slightly by further solving. It is only for the largest
cluster considered in the experiments, that it could be benefical to solve for
more than a few seconds. 

Next we investigate the effect of the integration of horizontal scaling with our
MILP, by comparing it with a non-integrated approach, which first 
performs scale-in as an independent process and then tries to move the key groups
from the to-be-removed nodes to the other nodes evenly. We experiment on the
largest cluster defined above. We set \textit{maxMigrations} to twenty and mark
ten nodes for removal. We tested two situations, where one and five nodes are
overloaded (100\% loaded), denoted as $1$OL and $5$OL respectively.  The result
in Figure~\ref{fig:integration} shows that the integrated approach obtains a
good load distance much faster than the non-integrated approach for both cases,
while being able to complete scaling-in within a similar time period. This is
because the integrated MILP can adaptively prioritize the more urgent migrations
to keep the load distance within a good range without sacrificing much on the
scaling time.

\begin{figure}[h] \centering
  \begin{subfigure}[b]{0.48\columnwidth} 
	\centering
	\includegraphics[height=80px]{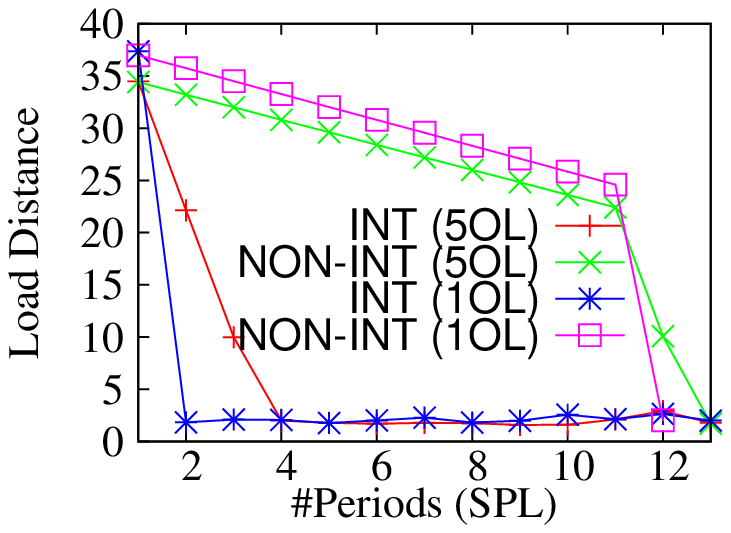}
	\caption{Load Distance}
\end{subfigure}%
\hfill
  \begin{subfigure}[b]{0.48\columnwidth} 
	\centering
	\includegraphics[height=80px]{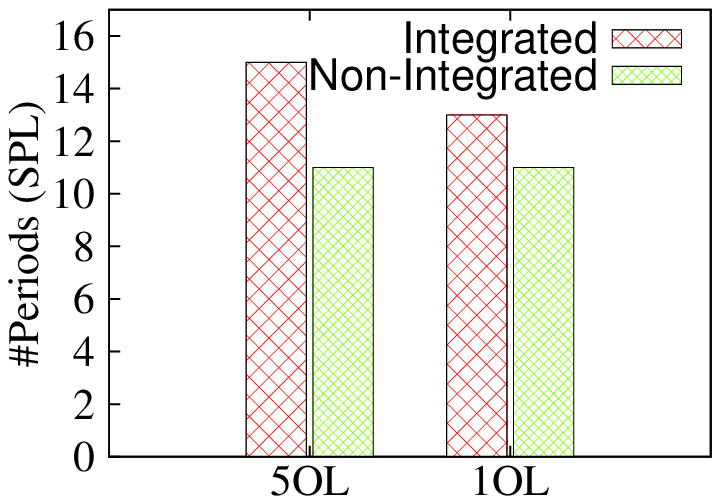}
	\caption{Time to scale in}
\end{subfigure}%
	\caption{Integrating horizontal scaling with load balancing}
	\label{fig:integration}
\end{figure}

\subsection{Load Balancing with MILP}
\label{sec:load balancing-experiment}
This experiment investigates the quality and overhead of load balancing which
can be obtained by solving the MILP. We first provide a comparison to Flux
\cite{shah:2003:flux} and to the ``Power of Both Choices'' (PoTC)
\cite{7113279}. Later, we evaluate the importance of limiting the overhead of
the MILP. The experiment was executed on $EC2$, with one \textit{input node} and
twenty \textit{worker nodes}. The experiment was executed on the dataset
\textit{Parsed Wikipedia edit history}.

\textbf{Real Job 1.} The job consists of one input and three operators with one hundred
key groups each. The first operator calculates a GeoHash value per input tuple.
The second operator calculates TopK updated articles with a window of $1$ minute
and the last operator calculates global TopK updated articles, also with a
window of $1$ minute. The dataset does not contain location data, so we assume a
completely even distribution of GeoHash values covering Denmark.

\begin{figure}[H] 
	\centering 
	\includegraphics[height=18px]{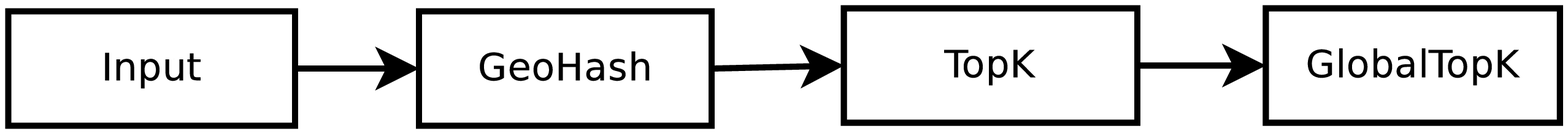}
\end{figure}

\subsubsection{Comparison with Flux and PoTC }
\label{eval:compare-flux}
We set $maxMigrations$ to 13 key groups per SPL, for both the MILP and Flux.

\begin{figure}[h] \centering
\begin{minipage}{0.48\columnwidth}
	\centering
	\includegraphics[width=110px]{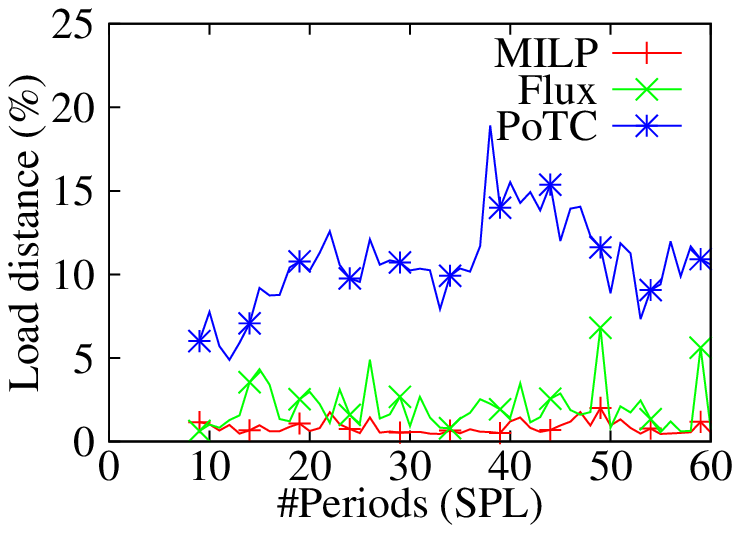}
	\caption{Quality}
	\label{fig:load_balance_performance}
\end{minipage}%
\hfill
\begin{minipage}{0.48\columnwidth}
	\centering
	\includegraphics[width=110px]{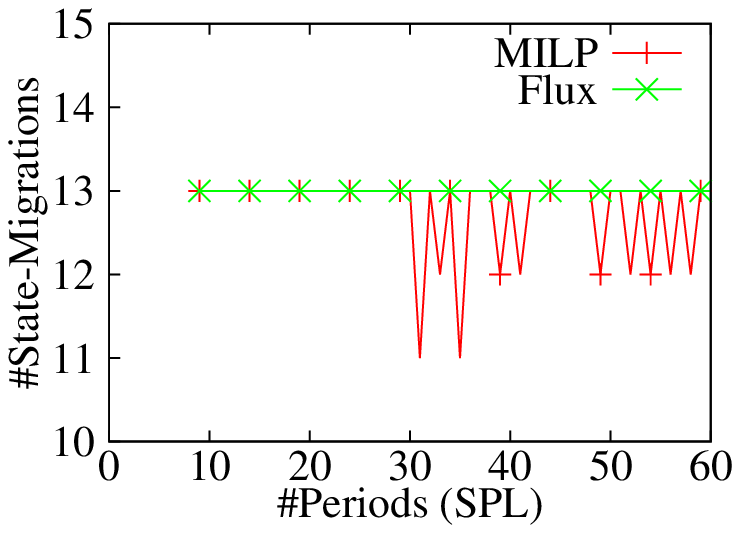}
	\caption{\#Migrations}
	\label{fig:load_balance_overhead}
\end{minipage}
\end{figure}


Figure \ref{fig:load_balance_performance} shows the load distances, directly after
applying migrations. The MILP approach achieves a stable load distance
consistently below 1\%, while Flux exhibits significant fluctuations up to 7\%.
The reason that our approach outperforms Flux, is that Flux makes sub-optimal
state migration decisions, i.e. Flux can require more state migration operations
to load balance than what is necessary, meaning the \textit{maxMigrations}
constraint will then prevent Flux from achieving as good load balance as the
MILP.

PoTC cannot achieve a consistently good load distance. This happens because the
job needs to do merge (every minute) and the amount of state to merge, varies
over time and from node to node. This introduces skewness in the load, which is
not considered by the PoTC approach.  Notice also that the PoTC approach incurs
a continuous overhead to process the merging operations, regardless of the load
distribution.



To recap, the benefits of our MILP is threefold: (1) it achieves
better load balancing than Flux and PoTC; (2) it outperforms Flux and
PoTC in terms of overhead; and (3) it can easily be adjusted by adding
other constraints to support specific cases, e.g. limiting the maximum memory
usage on a node.

\subsubsection{Unrestricted Load Balancing}
\label{subsubsec:unrestricted_load_balancing}
If the overhead of load balancing is not restricted, the MILP solver may
potentially migrate too many key groups. To avoid this problem, the MILP has a
constraint of the maximum migration cost which can be incurred per invocation.
In this experiment, we investigate how the load balance ``quality'' and the load
balance overhead, depends on this constraint (here shown as
\textit{maxMigrations}).

\begin{figure}[h]
\centering
\begin{minipage}{0.48\columnwidth}
	\centering
	\includegraphics[width=110px]{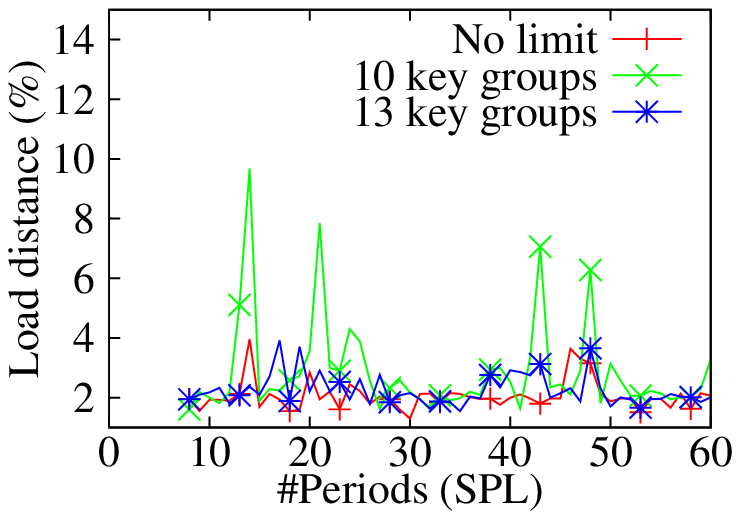}
	\caption{Quality}
	\label{fig:load_balance_unrestricted_performance}
\end{minipage}%
\hfill
\begin{minipage}{0.48\columnwidth}
	\centering
	\includegraphics[width=110px]{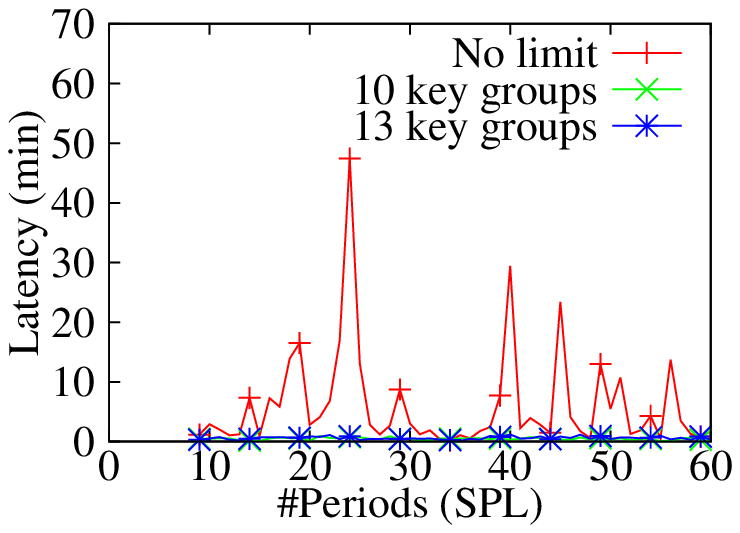}
	\caption{Overhead}
	\label{fig:load_balance_unrestricted_overhead}
\end{minipage}
\end{figure}

Figure \ref{fig:load_balance_unrestricted_performance} shows the obtained load
balancing is highly correlated with the number of allowed state migrations. The
unrestricted solution provides the best load balance, while the one with a limit
of 10 key groups provides a solution with significant spikes.
Figure~\ref{fig:load_balance_unrestricted_overhead} shows that the overhead in
terms of migration latency is very large for the unrestricted solution, due to a
large number of key groups being migrated. The y-axis is the sum of latency
incurred by all state migrations, which is defined as the amount of time that
the processing of a to-be-migrated key group is paused. 

Each key group to migrate incurs in average $2.5$ seconds of latency for this
experiment, when migrating up to $13$ key groups at a time. Assume e.g. $SPL$ is
set to five minutes, which means each of the 13 key groups to migrate, will only
be paused for $\frac{2.5}{300} \cdot 100 = 0.8\%$ of the total time. As the
experiment was executed with 300 key groups, only $\frac{2.5*13}{300 \cdot 300}
\cdot 100 = 0.04\%$ of the total processing incurs latency. Furthermore, using a
more sophisticated state migration technique than direct state migration, could
potentially reduce the latency to almost zero
\cite{Madsen:2015:DRM:2806416.2806449}. Our approach can therefore incur very
low overhead on the system.

\subsection{Load Balance and Collocation (Synthetic)}
\label{subsec:miqp_solver_performance}
In this experiment we use the synthetic data and job to compare ALBIC and COLA, in terms
of of load distance and
collocation. The synthetic data and job are more flexible for us to generate a
large topology and vary the parameters.  The experimental configuration here is
the same as that in Section~\ref{subsec:milp_solver_performance}, except that we
control the maximum obtainable collocation by ensuring x\% of the key groups to
have 1-1 communication. Furthermore, for each iteration of solving, the load of
20\% of the nodes is adjusted by a percentage, randomly chosen from the range
$[-2\%, 2\%]$.

\begin{figure}[h] \centering
\begin{minipage}{0.48\columnwidth}
	\centering
	\includegraphics[width=110px]{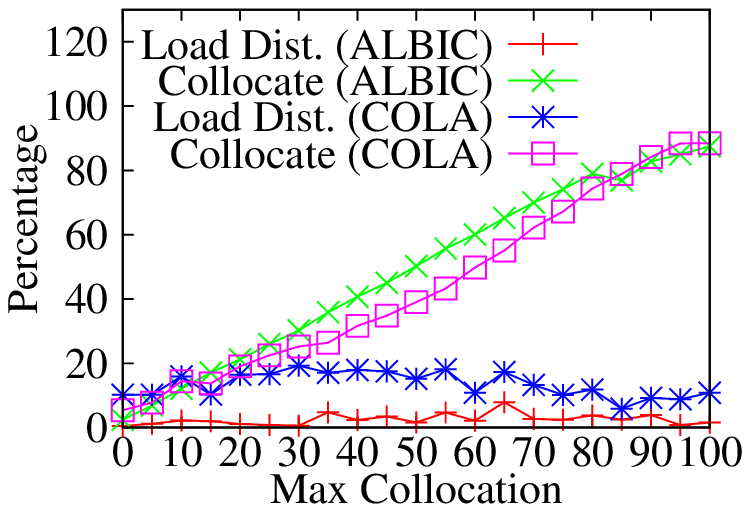}
	\caption{Collocation}
	\label{fig:miqp_solver_performance_2}
\end{minipage}%
\hfill
\begin{minipage}{0.48\columnwidth}
	\centering
	\includegraphics[width=110px]{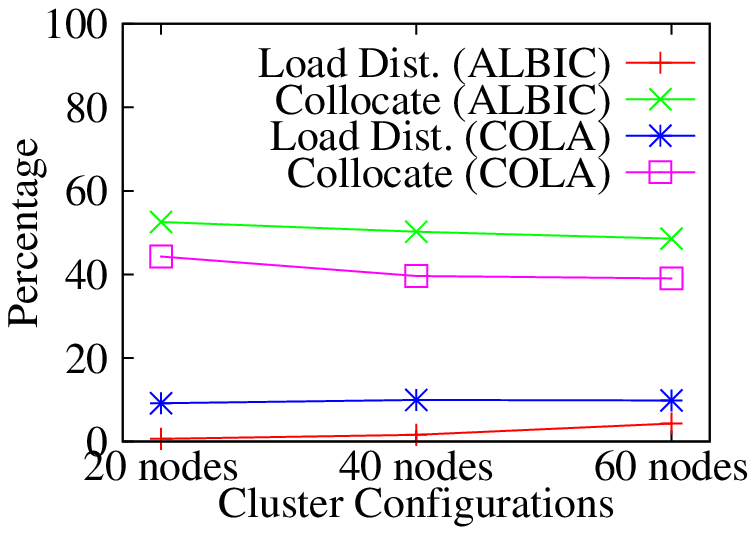}
	\caption{Configurations}
	\label{fig:miqp_solver_performance_3}
\end{minipage}
\end{figure}

%
\begin{figure*}[tb]
\centering
\begin{minipage}{4.38cm}
	\centering
	\includegraphics[width=110px]{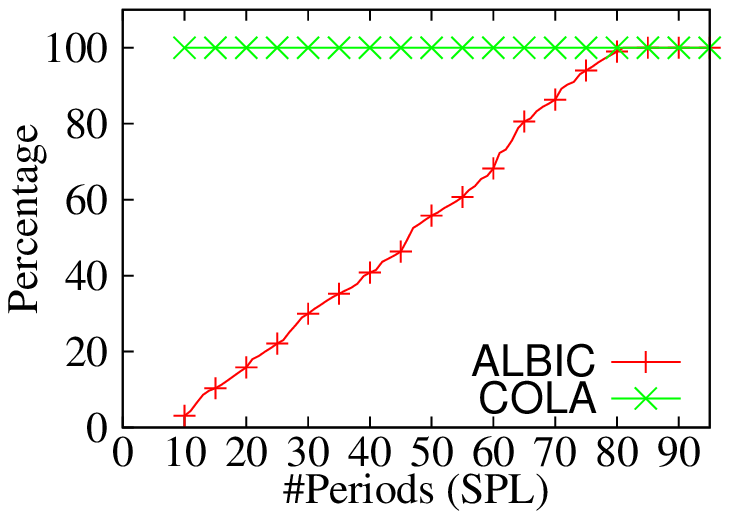}
	Collocation Factor
\end{minipage}%
\begin{minipage}{4.38cm}
	\centering
	\includegraphics[width=110px]{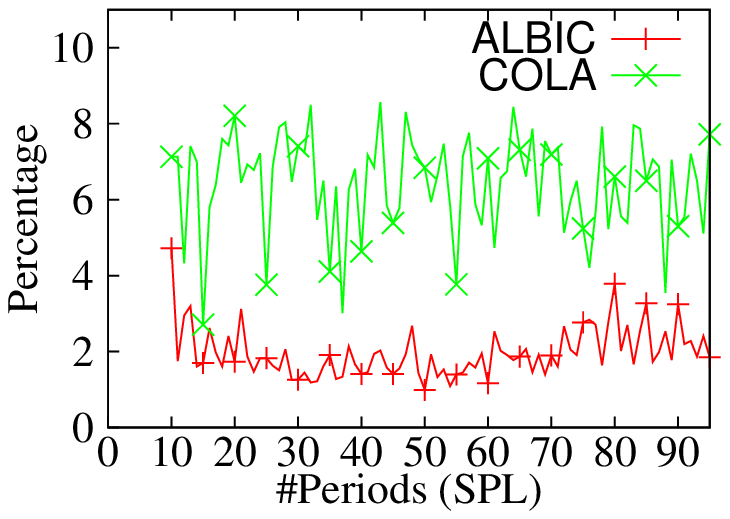}
	Load Distance
\end{minipage}
\begin{minipage}{4.38cm}
	\centering
	\includegraphics[width=110px]{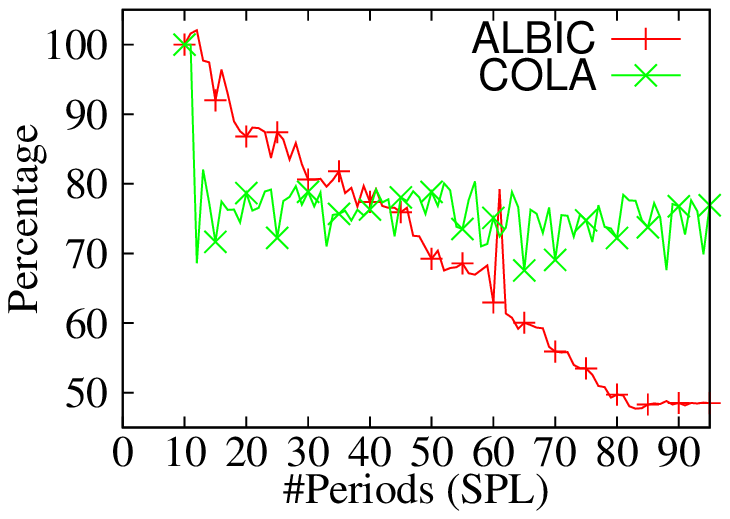}
	Load Index
\end{minipage}
\begin{minipage}{4.38cm}
	\centering
	\includegraphics[width=110px]{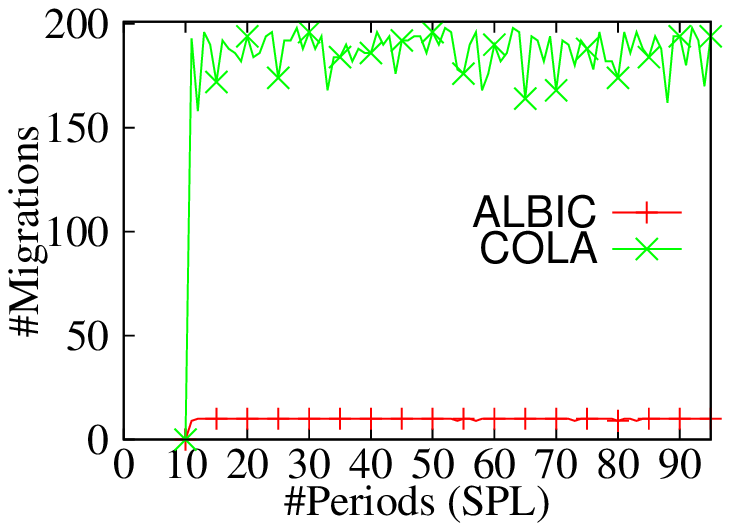}
	\#Migrations
\end{minipage}
\caption{Real Job 2}
\label{fig:mincost_lb_job1_2ops}
\end{figure*}

\begin{figure*}[tb]
\centering
\begin{minipage}{4.38cm}
	\centering
	\includegraphics[width=110px]{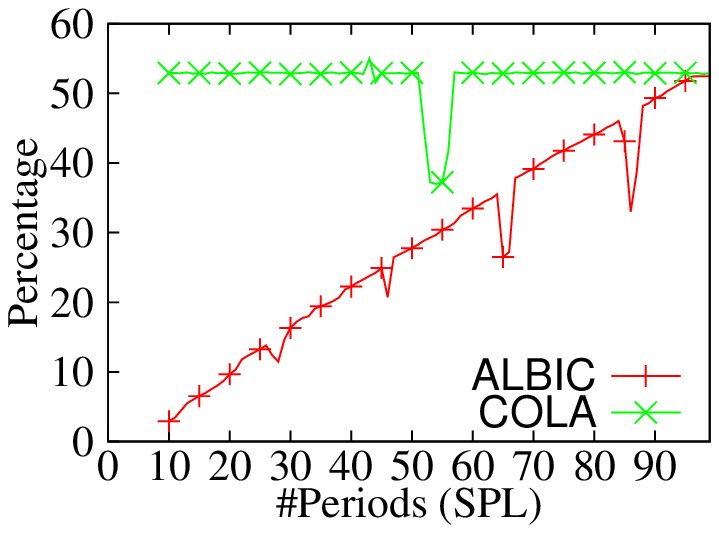}
	Collocation Factor
\end{minipage}%
\begin{minipage}{4.38cm}
	\centering
	\includegraphics[width=110px]{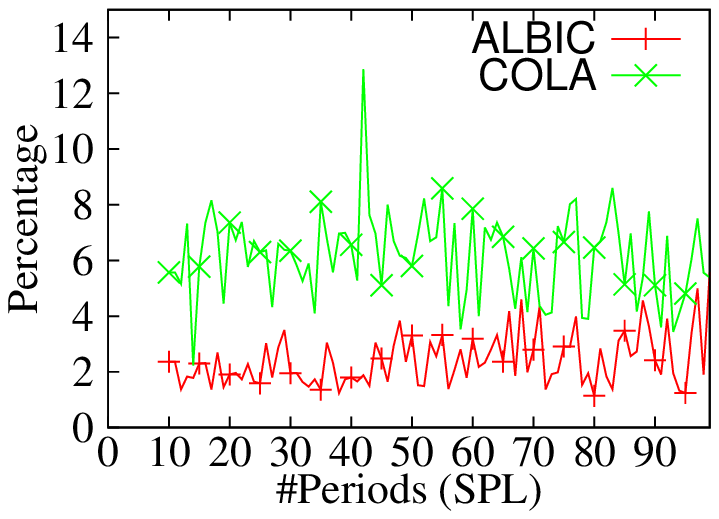}
	Load Distance
\end{minipage}
\begin{minipage}{4.38cm}
	\centering
	\includegraphics[width=110px]{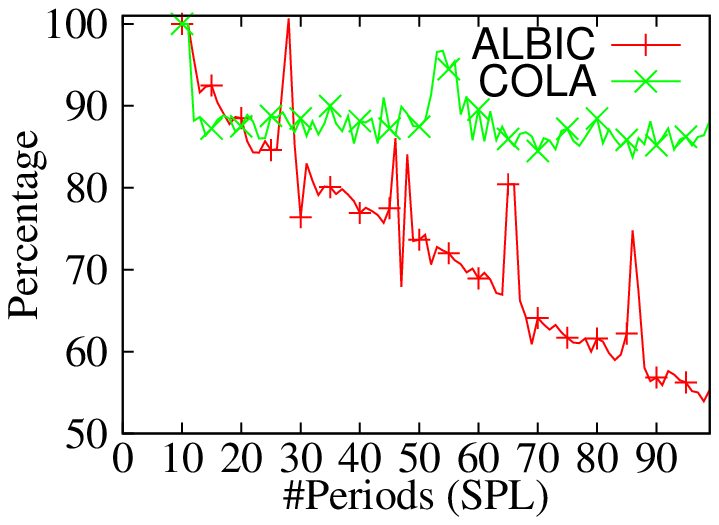}
	Load Index
\end{minipage}
\begin{minipage}{4.38cm}
	\centering
	\includegraphics[width=110px]{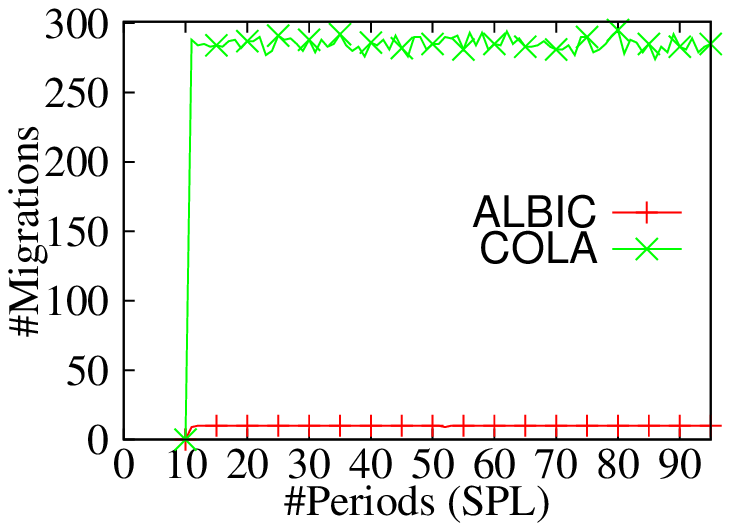}
	\#Migrations
\end{minipage}
\caption{Real Job 3}
\label{fig:mincost_lb_job2_3ops}
\end{figure*}

Figure~\ref{fig:miqp_solver_performance_2} considers a cluster of 40 nodes, 800
key groups and 20 operators. We set $maxMigrations = 20$ and vary the maximum
collocation factor from 0 to 100.
The figure shows that ALBIC achieves a smaller, hence better, load distance than
COLA. ALBIC also outperforms COLA in terms of collocation by up to ten percent.
Recall that both approaches try to define partitions containing collocated
key groups, while respecting some maximum load distance. The partitions must be
split until the load distance requirements are satisfied by the allocation.
Since ALBIC uses a much more sophisticated technique to do load balance, namely
MILP, in comparison to the simple heuristic in COLA, it does not need to split
partitions as much as COLA to achieve the same load distance. Therefore, it
achieves better key group collocations.

In Figure \ref{fig:miqp_solver_performance_3} we set the maximum collocation
factor to 50 and use the following three cluster
configurations: (1) 20 nodes, 400 key groups, 10 operators; (2) 40 nodes, 800
key groups, 20 operator; and (3) 60 nodes, 1200 key groups, 30 operators.
The results show that both ALBIC and COLA can achieve good solutions, while
ALBIC significantly and consistently outperforms COLA in both load
distance and collocation for various system sizes.


\subsection{Load Balance and Collocation (Real Data)}
\label{sec:collocation-experiment}

In this experiment, we compare ALBIC and COLA with real data and realistic jobs.
The experiment is executed on $EC2$, with four \textit{input node} and twenty
\textit{worker nodes}.

\textbf{Real Job 2} contains two operators, with five key groups per operator per
node. The first operator extracts delays and the second operator sums delays by
airplane per year. Both operators are parallelized on the same attribute. Thus,
it is possible to define a perfect collocation of operators, where no data needs
to be serialized and deserialized. Similarly, we can also define a worst
allocation of operators, where every tuple needs to be sent over the network.
The initial allocation of key groups is chosen such that the initial collocation
is as little as possible, which is to see if ALBIC can gradually increase the
collocation at runtime.

\begin{figure}[H]
	\centering
	\includegraphics[height=18px]{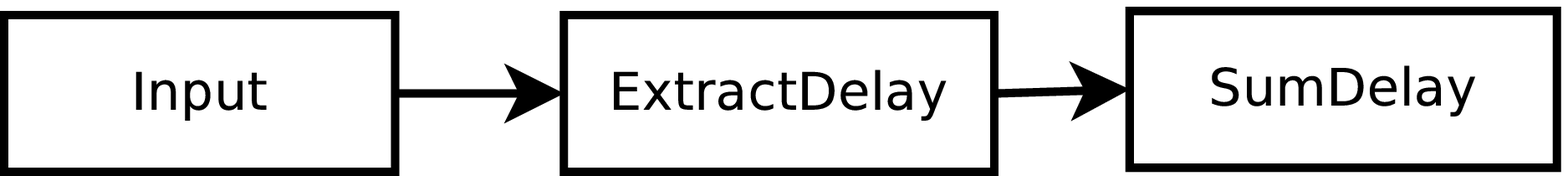}
\end{figure}

Figure~\ref{fig:mincost_lb_job1_2ops} shows the collocation factor, load distance,
load index and number of migrations done. As expected, COLA reaches the
optimum collocation factor immediately because it optimize the plan from
scratch. ALBIC, using its adaptive strategy, can gradually reach the same
collocation factor over time. Furthermore, thanks to the MILP solver, ALBIC can
consistently achieve a lower load distance than COLA, which adopts simple
heuristics. For the load index, ALBIC gets a decrease from $100\%$ to $50\%$,
which means the system load has been cut in half due to the collocation of
communicating key groups. COLA can only reduce the load by $25\%$, due to the
large overhead of migrations. COLA migrates close to 200 key groups per SPL,
while ALBIC only migrates 10 key groups. 
In summary, this experiment verifies that data communication does bring a
significant workload to the system and minimizing it helps rectifying system
overload, and potentially save the overhead of scaling-out.

{\bf Real Job 3} extends real job 2, with an additional operator that sums
delays per route, where a unique route is defined by having the same origin and
destination airport. We do not consider routes that span multiple airports.
To execute this experiment, the input rate for the COLA experiment is lowered to
50\% (maintaining data distribution), because the overhead of migrations is
simply too overwhelming for the system.

\begin{figure}[H]
	\centering
	\includegraphics[height=18px]{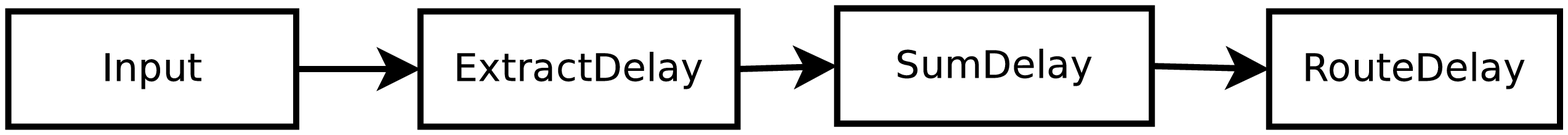}
\end{figure}


Figure~\ref{fig:mincost_lb_job2_3ops} shows a similar trend as the previous
experiment. One can see that the collocation factor is only
half of the previous experiment. This is because the RouteDelay operator cannot
be collocated with the SumDelay operator. 

{\bf Real Job 4} extends real job 3 with several additional operators. A
WeatherInput operator reads streaming input from the dataset \textit{Global
Surface Summary of the Day} based on the timestamps in the data. From the
weatherdata we calculate a rainscore, which is a value from 0 to 100, calculated
as the percentage of precipitation compared to the maximal historically measured
value. The higher the rainscore, the more rain there was. The computation could
be extended with a score for wind, atmospheric pressure and even thunderstorms.
Each route is joined with the rainscore for a given route and the courier
efficiency is calculated as the sum of delays for rainscores in intervals of
ten. The \textit{store} operators periodically writes results to a local
relational database.

\begin{figure}[H]
	\centering
	\includegraphics[height=40px]{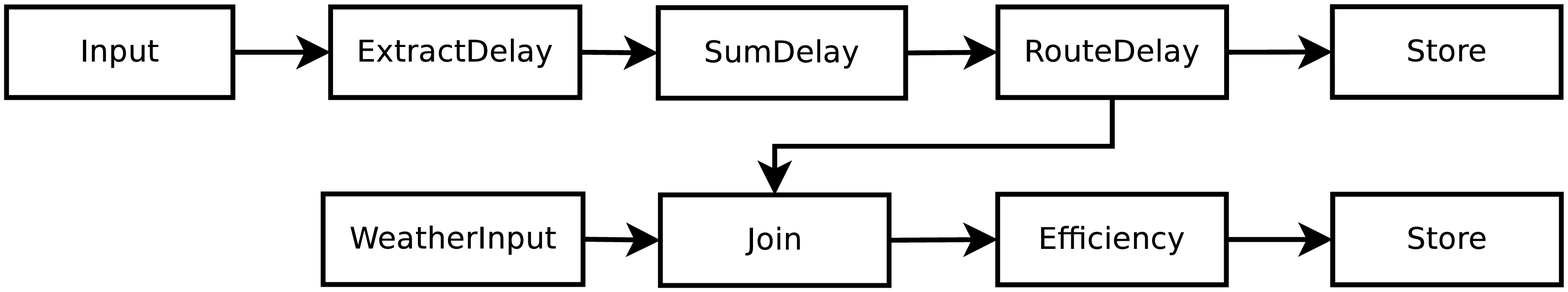}
\end{figure}

Figure \ref{result_real_job2_2} shows the results. For brevity, we do not show
the number of migrations conducted, but it is $10$ for each iteration of ALBIC
as expected.  It is impossible to run COLA for each adaptation period, since the
overhead of migrating key groups is too massive and exceeds the system capacity.
Instead, we have executed the job three times using a random allocation without
collocation and measured the collocation factor COLA achieves, which is very
consistent around 61\%. ALBIC gradually achieves a similar collocation factor and
reduces the load index while maintaining a low load distance. 

\begin{figure}[t]
	\centering
	\includegraphics[height=70px]{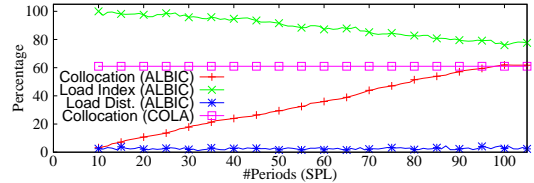}
	\caption{Real Job 4}
	\label{result_real_job2_2}
\end{figure}

As verified by the above experiments, both ALBIC and COLA can improve the
collocation, while exhibiting a reasonable load distance at runtime. The benefit
of ALBIC is twofold: (1) it incurs much less overhead at runtime and (2) it
continuously and adaptively minimizes the load distance. It would be reasonable
to use COLA for an initial key group allocation at job submission, and then to
use ALBIC for maintaining a good allocation at runtime. If one uses a simpler
load balancing algorithm such as MILP or Flux instead of ALBIC, the collocation
achieved by COLA would deteriorate at runtime.


We have also studied how collocation can improve the performance of Real Job 1
described in Section~\ref{sec:load balancing-experiment}. The result is that the
collocation maxes out at around 5\%, which is not large enough to make any
conclusions on the saved workload, as the savings are masked by the load
fluctuations. The collocation optimization has little effect to this job due to
the way the input data of each operator is partitioned.
The three partitioning functions are all independent on each other, hence they
all exhibit the Full Partitioning pattern with very even distributions, which
has little opportunity for collocation.

%
%

\section{Conclusion}
We have presented an integrated solution for load balancing, horizontal scaling
and operator instance collocation suitable for general PSPEs. We first
investigated how to model load balancing and dynamic scaling as an MILP, which
is suitable when collocation of operator instances has little
effect on the communication cost. As verified by our experiments, solving the
MILP problem can achieve better load balance with a lower overhead of state
migration compared to existing approaches. We then presented an extension of
MILP called ALBIC, which optimizes the collocation of operator instances, while
maintaining good load balancing and incurring low overhead at runtime.  Our
experiments verify that it maximizes the beneficial collocations without
sacrificing the load balance and adaptation cost.

\balance

\bibliographystyle{abbrv}
\bibliography{document}
%

\end{document}